\documentclass[conference, 10pt, twocolumn]{ieeeconf}

\IEEEoverridecommandlockouts
\usepackage[usenames]{color}
\usepackage{enumerate}
\usepackage{url}
\usepackage{subfigure}
\usepackage{amsfonts,mathrsfs}
\usepackage{verbatim}
\usepackage{acronym}
\usepackage{pifont}
\usepackage{graphicx}
\usepackage{amsthm}
\usepackage{ifmtarg}
\usepackage{algorithm}
\usepackage{algpseudocode}
\usepackage{mathtools}
\usepackage{cite}
\usepackage{lipsum}
\usepackage{stfloats}

\def\fskip#1{}

\newtheorem{theorem}{Theorem}

\newtheorem{assumption}{Assumption}

\newtheorem{definition}{Definition}

\newtheorem{lemma}{Lemma}

\newtheorem{remark}{Remark}

\renewcommand{\theenumi}{\arabic{enumi}}

\def\1{{\bf 1}}

\newcommand{\remove}[1]{}

\newcommand*{\TitleFont}{%
      \usefont{\encodingdefault}{\rmdefault}{}{n}%
      \fontsize{19}{18}%
      \selectfont}

\def\argmin{\mathop{\rm argmin}}

\makeatletter
\algrenewcommand\ALG@beginalgorithmic{\footnotesize}
\makeatother

\begin{document}
\title{{\TitleFont \vspace{-1cm}Stochastic Games for Smart Grid Energy Management with Prospect Prosumers}}\vspace{-1cm}
\author{\authorblockN{S. Rasoul Etesami,\ Walid Saad,\ Narayan Mandayam,\ and H. Vincent Poor\vspace{-1cm}}\vspace{-2cm}\\
\thanks{S. Rasoul Etesami and H. Vincent Poor are with Department of Electrical Engineering, Princeton University, email: (setesami,poor)@princeton.edu.}
\thanks{Walid Saad is with Wireless@VT, Department of Electrical and Computer Engineering, Virginia Tech, Blacksburg, VA USA (email: walids@vt.edu).}
\thanks{Narayan Mandayam is with WINLAB, Department of ECE, Rutgers University,
North Brunswick, NJ 08902, (email: narayan@winlab.rutgers.edu).}
\thanks{This research was supported by the NSF under Grants  ECCS-1549881, ECCS-1549900, CNS-1446621, and ECCS-1549894.}}
\maketitle
\thispagestyle{empty}
\pagestyle{empty}  

\vspace{-2cm}
\begin{abstract}
In this paper, the problem of smart grid energy management under stochastic dynamics is investigated. In the considered model, at the demand side, it is assumed that customers can act as prosumers who own renewable energy sources and can both produce and consume energy. Due to the coupling between the prosumers' decisions and the stochastic nature of renewable energy, the interaction among prosumers is formulated as a stochastic game, in which each prosumer seeks to maximize its payoff, in terms of revenues, by controlling its energy consumption and demand. In particular, the subjective behavior of prosumers is explicitly reflected into their payoff functions using prospect theory, a powerful framework that allows modeling real-life human choices, rather than objective, user-agnostic decisions, as normative models do. For this prospect-based stochastic game, it is shown that there always exists a stationary Nash equilibrium where the prosumers' trading policies in the equilibrium are independent of the time and their histories of the play. Moreover, to obtain one of such equilibrium policies, a novel distributed algorithm with no information sharing among prosumers is proposed and shown to converge to an $\epsilon$-Nash equilibrium in which each prosumer is able to achieve its optimal payoff in an equilibrium up to a small additive error $\epsilon$. On the other hand, at the supply side, the interaction between the utility company and the prosumers is formulated as an online optimization problem in which the utility company's goal is to learn its optimal energy allocation rules. For this case, it is shown that such an optimization problem admits a no-regret algorithm meaning that regardless of the actual outcome of the game among the prosumers, the utility company can follow a strategy that mitigates its allocation costs as if it knew the entire demand market a priori. Simulation results show the convergence of the proposed algorithms to their predicted outcomes and present new insights resulting from prospect theory that contribute toward more efficient energy management in the smart grids. 
\end{abstract} 

\section{Introduction}

The electric power grid is evolving into a heterogeneous smart grid econsystem that will seamlessly integrate renewable resources, storage units, electric vehicles, smart meters, and other intelligent appliances \cite{lakshminarayana2014cooperation,atzeni2013demand,maharjan2013dependable}. One key feature of the smart grid is the use of storage devices for energy management among grid components. As shown recently in \cite{urgaonkar2011optimal,vytelingum2010agent}, incorporating storage devices into the grid design can significantly improve energy management and result in huge cost saving in electricity delivery. In fact, the role of storage devices is even more pronounced when a portion of injected electricity to the grid is obtained from renewables (e.g., wind power or solar energy). This is due to the many uncertainties, such as weather conditions, that can impact the amount of generated energy from such resources. In such situations, having access to storage units to save the current excess energy and use it whenever there is energy shortage in the grid will bring a lot of flexibility into the energy management. Therefore, managing uncertainties using storage devices and properly controlling the production and distribution of electricity are some of the most important challenges in the design and analysis of smart grids.

\vspace{-0.2cm}
\subsection{Related Work} 
 
There has been significant recent works that investigated the challenges of grid energy management in the presence of storage units and renewable energy \cite{wang2015load,wang2012dynamic,bayram2014survey,liu2014pricing,huang2002design,atzeni2012day}. In \cite{wang2015load} the authors studied the problem of energy management using load shifting thus allowing a part of the peak hour load to be moved to
an off-peak hour. However, their formulation was based on a static noncooperative game which cannot capture the stochastic nature of various grid components, such as renewable energy. Moreover, the work in \cite{wang2015load} relies on the notion of weighting from prospect theory, however, in practice both weighting and framing effects can impact consumer behavior (as will be shown in this work). In \cite{wang2012dynamic} an optimization framework was adopted to study the price fluctuations in electricity market. However, in this work, the agents were assumed to be either producers or consumers of energy. The works in \cite{huang2002design,bayram2014survey}, and \cite{saad2011noncooperative} studied grid energy management using an incentive compatible double-auction mechanism where the grid users bid for energy and utility companies set the price. However, these works focus on mechanism design problems and do not account for the behavior of consumers. A prediction-based pricing mechanism for data centers was proposed in \cite{liu2014pricing} while taking into account renewable energy sources. However, the work in \cite{liu2014pricing} relies on centralized optimization problem rather than on a distributed game-theoretic framework. Moreover, the work in \cite{atzeni2012day} proposed a day-ahead bidding strategy which allows the supply-side to know in advance an estimate of the amount of energy to be provided to the demand-side during the upcoming day. However, the formulation in \cite{atzeni2012day} is based on a static noncooperative game and it does \emph{not} take into account the subjective behavior of the users.

The effects of integrating storage units on the resilience of smart grids has been investigated in \cite{rahi2016prospect}. Moreover, controlling the uncertainties of renewable resources using stochastic optimization methods was studied in \cite{lakshminarayana2014cooperation} and \cite{chen2013energy}. Since the solution of such optimization problems is challenging, there has been an increased interest in the use of other methods such as online learning or regret minimization with the aim of handling uncertainties as done in \cite{kim2014real} and \cite{rose2011learning}. In addition, there has been a number of works that appeared on the use of game theory for analyzing the interactions between consumers and power companies \cite{saad2012game,maharjan2013dependable}, and \cite{wang2014game} mainly based on Stackelberg or static game models. More recently, it has been observed that in practical grids the subjective view of decision makers about their opponents can play an important role in changing the final outcomes \cite{saad2016toward}. This new perspective is often inspired from prospect theory (PT) \cite{kahneman1979prospect}, a mathematical framework which explains real-life decision-making and its deviations from conventional game theory. Indeed, as shown in \cite{wang2015load,xiao2015prospect}, and \cite{rahi2016prospect}, explicitly accounting for the subjective behavior of prosumers using PT can substantially change the anticipated energy management outcomes resulting from the participation of prosumers in energy trading. However, most of these works focus on the \emph{weighting} effects of PT, and ignore the \emph{framing} effects.

Despite being interesting, these existing works have focused on static game formulations which are not quite descriptive under practical smart grid settings in which many factors, such as renewable energy, are highly dynamic and stochastic. This motivates the use of a richer class of games, namely \textit{stochastic games}, in order to capture such dynamic uncertain environments. Stochastic games have been extensively studied in the literature \cite{fink1964equilibrium,altman2008constrained,basar1978decentralized}, with solutions crucially based on reinforcement learning \cite{bowling2000analysis}; however, beyond a handful of works \cite{weistochastic} that address the security of power grids against cyber-physical attacks, most of the existing applications do not address the problem of smart grid energy management. In particular, in most of the literature on stochastic games \cite{fink1964equilibrium,weistochastic}, a typical assumption is that each of the players has full information about the entire underlying game, which is a questionable assumption in many practical settings. For instance, in many situations, the prosumers do not have access to each others' private information such as the amount of stored or consumed energy. 

\vspace{-0.2cm}
\subsection{Contributions}

Our work in this paper provides one of the first formulations for smart grid energy management using stochastic games which works under incomplete information settings. Furthermore, since many grid components are owned and operated by humans, their subjective perceptions and decisions can substantially affect the grid outcomes. This makes PT a natural choice for smart grid design and analysis under real behavioral considerations on the behavior of smart grid prosumers. In particular, the subjective behavior of prosumers is even more pronounced when their uncertainty about the grid (as a result of stochastic renewables) or their opponents' decisions increases. Therefore, studying real-life decision-makings in smart grids under highly dynamic settings using a stochastic game framework has not been addressed before.

The main contribution of this paper is to develop a novel framework for smart grid energy management which takes into account the stochastic nature of renewable energy, the distributed nature of the system, and the subjective perceptions of the prosumers. Our work differs from most existing literature \cite{wang2015load,wang2012dynamic,bayram2014survey,liu2014pricing,huang2002design,atzeni2012day} in several aspects: 1) it models interactions between selfish prosumers using a stochastic game with incomplete information, 2) it incorporates real-life decision behavior of prosumers under the stochastic game framework (rather than a static game such as in \cite{wang2015load}) by using the framing and weighting effects of PT and studies its deviations from conventional expected utility theory (EUT), 3) it provides a novel distributed algorithm to obtain equilibrium points of the system with incomplete information (unlike conventional methods such as reinforcement learning \cite{la2016cumulative,lin2013stochastic}), and 4) it provides a computationally tractable and robust formulation for the utility company using online convex optimization.

In the studied model, at the demand side, we consider a set of prosumers which can both produce and consume energy. At each time instant, each prosumer decides on the amount of energy to consume and the amount of extra energy demand to buy from the utility company under a given pricing rule. The energy demand of each prosumer can depend on the energy storage level as well as the generated renewable energy of all other prosumers. In particular, each prosumer makes a decision based on its own subjective view on the energy market outcomes in order to maximize its payoff which reflects its loss or gain, in terms of revenues, in the grid. We show that such a game admits a stationary Nash equilibrium (NE) and we propose a distributed algorithm on how to drive the entire market to an $\epsilon$-NE, i.e., to a point which nearly maximizes all the prosumers' payoffs. Given these prosumer decisions, and in order to satisfy their energy demands, we further consider a utility company which responds them by allocating energy to different substations, and whose goal is to predict prosumers' demands in order to satisfy their energy needs more reliably while minimizing its own energy distribution costs. Simulation results illustrate the convergence of the proposed algorithm and show that the optimal decisions of prosumers in the case of PT are considerably different from those resulting from a conventional game. As an example, prosumers who have higher subjective evaluations tend to consume even more energy in the lower range of their energy storage, and they become more conservative as their energy storage levels increase. Such unusual consumption patterns will potentially require new benchmarks on how to distribute and manage energy more efficiently across power grids in which prosumers are actively participating in energy management.


The rest of the paper is organized as follows. In Section \ref{sec:model}, we introduce our system model formally. In Section \ref{sec:NE}, we prove the existence of an NE in the stochastic game among prosumers. We provide a learning algorithm for finding an $\epsilon$-NE policy among prosumers in Section \ref{sec:learning-prosumers}. In Section \ref{sec:online-convex}, we describe a no-regret algorithm for the utility company. Simulation results are given in Section \ref{sec:simulation}. We conclude the paper in Section \ref{sec:conclusion}.
 
{\bf Notation}: For a real number $x$, we let $[x]^+=\max\{0,x\}$.  For a vector $\boldsymbol{v}$, we let $v_i$ be its $i$th component and $\boldsymbol{v}'$ be its transpose. We denote all but the $i$th component of a vector $\boldsymbol{v}$ by $\boldsymbol{v}_{-i}$. We let $\mathbb{I}_{\mathcal{A}}$ be the indicator function of a set $\mathcal{A}$. Finally, we let $\Pi_{\mathcal{A}}[\cdot]$ be the projection operator on a closed convex set $\mathcal{A}$, i.e., $\Pi_{\mathcal{A}}[x]=\argmin_{y\in\mathcal{A}}\|x-y\|$.




\section{System Model and Problem Formulation}\label{sec:model}

In this section we provide our problem formulation for energy management in smart grid. On the supply side, there is a utility company which interacts with a set of prosumers by selling electricity. On the demand side, there are many prosumers who interact with each other as well as with the utility company through a non-cooperative stochastic game. In Subsection \ref{sec:stochastic-game}, we first describe the interaction game among prosumers, and postpone the formulation for the utility company until Subsection \ref{sec:utility-company}.
 
\subsection{Stochastic Game Among Prosumers}\label{sec:stochastic-game}
We consider a set $\mathcal{N}:=\{1,2,\ldots,N\}$ of $N$ prosumers, each of which can both produce and consume energy. Each prosumer $i\in \mathcal{N}$ has a storage unit of maximum capacity $S^{\max}_{i}$, whose energy level at time step $t=1,2,\ldots$ is modeled by a random variable $S_i(t)$ (note that $\{S_i(t)\}_{t=1}^{\infty}$ is a random process). We assume that each storage is subject to self-discharge\footnote{Self-discharge is due to internal chemical reactions, just as closed-circuit discharge is, and tends to occur more quickly at higher temperatures \cite{wu2000self}.} which is a common phenomenon in batteries and reduces the stored charge without any connection. Moreover, each prosumer is equipped with a renewable energy generator, such as a wind turbine or a solar panel. We denote the effective generated energy, or simply \textit{generated} energy, which is the energy harvested from renewable resources minus the storage self-discharge of prosumer $i$ at time $t$ by random variable $G_i(t)$. Since the amount of generated energy depends on many random factors such as climate conditions, therefore, in general $G_i(t)$ is a random variable which admits both negative or positive values (negative when the battery self-discharge is more than the harvested energy, and positive, otherwise). 

We denote the amount of energy consumed by prosumer $i$ between times $t$ and $t+1$ by $L_i(t)$. Moreover, we assume that each prosumer can compensate its shortage of energy by buying additional energy (if needed) from the utility company. We denote the extra energy demanded by prosumer $i$ from the utility company at time $t$ by $D_i(t)$. Here, consumption refers to any type of energy usage by a prosumer such as energy used for lighting or heating, while demand is the amount of additional energy requested by a prosumer from the utility company in order to satisfy its needs. Consequently, the storage level of prosumer $i$ in the next time step, $S_i(t+1)$, will be:
\begin{align}\label{eq:state-transition}
S_i(t+1)\!=\!\min\Big\{[S_i(t)\!+\!G_i(t)\!+\!D_i(t)\!-\!L_i(t)]^+, S^{\max}_{i}\Big\}. 
\end{align}

Due to the fact that, in reality, the amount of traded energy, stored energy, or price are measured based on discrete quantities (e.g. 1kWh, \$1) even though the nature of these parameters is continuous, we let the range of generated energy, consumed energy, demanded energy, and stored energy of prosumer $i$ be discrete sets $\mathcal{G}_i:=\mathbb{Z}$, $\mathcal{L}_i:=\{0,1,\ldots,L_i^{\max}\}$, $\mathcal{D}_i:=\{0,1,\ldots,D_i^{\max}\}$, and $\mathcal{S}_i:=\{0,1,\ldots,S_i^{\max}\}$, respectively.  In fact, the storage level of each prosumer can be considered as its \textit{state} at time $t$ which evolves according to a stochastic process governing the randomness of generated energy from renewable resources. Thus, depending on how much energy is left in the storage until from time $t$, $S_i(t)$, and the amount of generated energy at that time $G_i(t)$, prosumer $i$ must take an \textit{action} $A_i(t):=(L_i(t),D_i(t))$, so as to determine how much energy to consume and how many additional energy units to demand from the utility company, in which case its \textit{actual instantaneous} payoff will be
\begin{align}\label{eq:prosumer-payoff}
U_i(A_i(t),\boldsymbol{A}_{-i}(t))=f_i(A_i(t))-D_i(t)\times p_i(\boldsymbol{D}(t)),
\end{align}where $f_i(\cdot)$ is an increasing function representing the satisfaction of prosumer $i$ from consuming $L_i(t)$ units of energy, and $p_i(\cdot)$ is the energy price function based on all prosumers' demands which is charged to prosumer $i$ for buying $D_i(t)$ units of energy given other prosumers' demand $\boldsymbol{D}_{-i}(t)$. Note that $D_i(t)$ can be a function of $L_i(t)$ and $S_i(t)$, where by \eqref{eq:state-transition} the latter itself is a function of the random generated energy $G_i(t)$. Thus, the payoff of each player is implicitly a function of generated energy and storage level of all prosumers. 
\begin{remark}
Although $D_i(t)$ and $L_i(t)$ are assumed to belong to independent sets $\mathcal{D}_i$ and $\mathcal{L}_i$, however, they are constrained by the state dynamics \eqref{eq:state-transition}. Nevertheless, our analysis remain unchanged even if $A_i(t)$ belongs to a more general constraint set, e.g., $A_i(t)\in \{(L_i,D_i)\in \mathcal{L}_i\times \mathcal{D}_i: L_i+D_i\leq \Delta\}$.
\end{remark}

We now provide a formal definition on how the prosumers can select their actions at different time instants \cite{altman1999constrained}: 

\begin{definition}\label{def:policy}
\normalfont A \emph{policy} $\Psi_i$ for prosumer $i$ is a sequence of probability measures $\Psi_i:=\{\Psi_i(t),t=1,2,\ldots\}$ over the action set $\mathcal{A}_i$ such that at each time $t$ chooses an action from $\mathcal{A}_i:=\mathcal{L}_i\times\mathcal{D}_i$ according to the probability measure $\Psi_i(t)$ whose distribution in general can be a function of past histories of states, actions, or even time. A policy $\Psi_i$ is called a \emph{stationary policy} if the probability of choosing an action $a_i\in \mathcal{A}_i$ only depends on the current state $s_i$ and is independent of the time $t$. In the case of stationary policy, we denote this time independent probability by $\Psi_i(a_i|s_i)$.  
\end{definition}
Next, assume that each prosumer $i\in \mathcal{N}$ chooses its action at different time instants based on some policy $\Psi_i$. Here, we note that due to the discrete nature of variables, the instantaneous payoff received by prosumer $i$ is a discrete random variable whose randomness comes from two different sources: i) The random generated energy from renewable resources. Indeed, as seen from \eqref{eq:state-transition}, the state of prosumer $i$ evolves as a function of random process $\{G_i(t), t=1,2,\ldots\}$ and ii) The internal randomness of joint policies $\boldsymbol{\Psi}=(\Psi_1,\ldots,\Psi_n)$. This is due to the fact that the prosumers choose their actions at different stages based on some probability distribution which depends on the joint policies. As a result, although each prosumer has knowledge about its own policy and random generated energy, however, it is quite uncertain about the randomness caused by the \emph{other} prosumers' decisions and generated energy. 

In this regard, there is strong evidence \cite{kahneman1979prospect} that, in the real-world, human decision makers do not make decisions based on expected values of outcomes evaluated by actual probabilities, but rather based on their perception on the potential value of losses and gains associated with an outcome. Indeed, using prospect theory (PT), the authors in \cite{kahneman1979prospect} showed that most people will often overestimate low probability outcomes and underestimate high probability outcomes. This phenomenon, known as \textit{weighting} effect in PT, reflects the fact that humans usually have subjective views on uncertain outcomes. For instance, for a given policy $\Psi_i$, prosumer $i$'s payoff  depends with some probability on others' policies $\boldsymbol{\Psi}_{-i}$ and renewable resources $\boldsymbol{G}_{-i}$. However, due to the uncertainty that prosumer $i$ has about its opponents' policies and renewables, instead of using actual probabilities induced by $\boldsymbol{\Psi}_{-i}$ and $\boldsymbol{G}_{-i}$ to evaluate its own expected payoff and then take its action accordingly, prosumer $i$ may use a weighted version of those using a nonlinear function $w_i(\cdot)$, taking into account its own subjective view about its opponents.\footnote{For example, prosumers often perceive losses more than gains and intend to overweight their losses and underweight their gains.} Moreover, each prosumer may have different perception about its loss or gain. This phenomenon that in reality human sort their losses or gains with respect to a reference point using their own, individual and subjective value function is known as \textit{framing} effect in PT. This differs from conventional expected utility theory (EUT),\footnote{In EUT the decisions are made purely based on conventional expectation.} which assumes players are rational agents that are indifferent to the reference point with respect to which their losses or gains are evaluated. Hence, at each time instant $t$, each prosumer receives a payoff which is the realization of a random
variable whose expectation equals its expected prospect payoff and, then, this prosumer makes its next decision.


To capture such human decision behavior, we use the following definition from PT \cite{kahneman1979prospect}:
\begin{definition}\label{def:prospect}
\normalfont Any prosumer $i$ has two corresponding functions $w_i(\cdot):[0,1]\to \mathbb{R}$ and $v_i(\cdot):\mathbb{R}\to \mathbb{R}$, known as \emph{weighting} and \emph{valuation} functions. The \emph{expected prospect} of a random variable $X$ with outcomes $x_1,x_2,\ldots,x_k$, and corresponding probabilities $p_1,p_2,\ldots,p_k$, for prosumer $i$ is given by $\mathbb{E}^{^{\rm PT}}[X]:=\sum_{\ell=1}^{k}w_i(p_{\ell})v_i(x_{\ell})$.
\end{definition}


Although there are many weighting and valuation functions, however, two of the widely used functions in the literature are known as Prelec weighting function and Tversky valuation function defined by \cite{prelec1998probability,al2008note},
\begin{align}\label{eq:weight-value}
&w(p)=\exp(-(-\ln p)^{c}),\cr 
&v(x)=\begin{cases} x^{c_1} & \mbox{if} \ \ x\ge 0, \\ 
-c_2(-x)^{c_3} & \mbox{if} \ \ x<0,\end{cases}
\end{align}
where $0<c\leq 1$ is a constant denoting the distortion between subjective and objective probability, and $c_1,c_2,c_3>0$ are constants denoting the degree of loss aversion. In fact, the functions in \eqref{eq:weight-value} are suggested based on extensive real-world experiments. However, our approach can accommodate any type of such functions without making use of the specific forms given in \eqref{eq:weight-value} as long as the weighting function $w_i(\cdot)$ is a continuous function and $w_i(\cdot)$ and $v_i(\cdot)$ satisfy axioms of valid weighting and valuation functions in PT \cite{kahneman1979prospect} (we only use the specific functions \eqref{eq:weight-value} in Section \ref{sec:simulation} to provide more concrete simulation results).  

Therefore, based on Definition \ref{def:prospect}, we assume that each prosumer $i$ makes it decision at time $t$ based on the realization of a random variable whose expectation equals its expected \emph{prospect} payoff $\mathbb{E}^{^{\rm PT}}_{\boldsymbol{\Psi},\boldsymbol{G}}[U_i(\boldsymbol{A}(t))]$, rather than its actual expected payoff. Here, the expectation is with respect to the internal randomness induced by joint policy of all prosumers, $\boldsymbol{\Psi}$, as well as their random generated energy processes $\boldsymbol{G}$. As the prosumers are individually maximizing their own payoffs and since their actions are coupled, the use of a game-theoretic solution \cite{basar1999dynamic} is apropos. Hence, we can formally formulate a stochastic game among prosumers as follows:


\begin{itemize}
\item A set of prosumers (players) $\mathcal{N}$. Each prosumer $i\in \mathcal{N}$ has an action set $\mathcal{A}_i$, and a state set $\mathcal{S}_i$.
\item Denoting the entire state of the game at time step $t$ by $\boldsymbol{S}(t)\in \mathcal{S}_1\!\times\!\ldots\!\times\!\mathcal{S}_N$, each prosumer $i$ takes an action $A_i(t)\in \mathcal{A}_i$ with some probability based on its own policy $\Psi_i$ and receives a payoff which is the realization of a random variable whose expectation equals its expected prospect payoff $\mathbb{E}^{^{\rm PT}}_{\boldsymbol{\Psi},\boldsymbol{G}}[U_i(\boldsymbol{A}(t)]$. 
\item Depending on what states and actions are realized at time $t$, the state of the game will move to a subsequent random state $\boldsymbol{S}(t+1)$ whose distribution depends on the probabilities that the players choose their actions and those of being in different states at time $t$. 
\item Since the players choose their actions at different stages of the game based on some policy profile $\boldsymbol{\Psi}:=(\Psi_i,\boldsymbol{\Psi}_{-i})$, the average expected payoff received by prosumer $i$ is given by    
\begin{align}\label{eq:discounted-payoff}
V_i(\Psi_i,\boldsymbol{\Psi}_{-i}):&=\lim_{T\to \infty}\frac{1}{T} \sum_{t=1}^{T}\mathbb{E}^{^{\rm PT}}_{\boldsymbol{\Psi},\boldsymbol{G}}\Big[U_i\big(\boldsymbol{A}(t)\big)\Big],
\end{align}
where the expectation is with respect to the randomness induced by the joint policy $\boldsymbol{\Psi}$, and state of the game. \footnote{If the limit does not exist, we can replace $\lim$ with $\limsup$. However, as we will see, for stationary policies this limit always exists.} 
\end{itemize}

In this game, each prosumer in the grid seeks to select a policy which maximizes its own average perceived revenue given by \eqref{eq:discounted-payoff}.  

\subsection{Optimization Problem for the Utility Company}\label{sec:utility-company}

Next, we incorporate the role of utility company into our problem setting. Using \eqref{eq:prosumer-payoff}, one can see that the utility company interacts with prosumers in two different ways: i) Setting the electricity prices and ii) Producing enough energy and distributing it among prosumers to satisfy their demands. Although one can consider the utility company itself as an additional player interacting with prosumers, however, due to the fact that, independent of the prosumers' actual demands, the utility company is responsible to satisfy their needs by generating and distributing enough energy among them, we formulate the utility company's problem as a separate optimization problem. Since generating and distributing energy can be very costly for the utility company, hence for a fixed pricing rule $\boldsymbol{p}(\cdot):=(p_1(\cdot),\ldots,p_n(\cdot))$, the most important question is to devise the amount of energy to produce and the way to distribute it in the grid. In other words, the utility company's goal is to learn the demand functions of the prosumers, in which case it can generate and allocate accurate energy units to each prosumer, and hence, minimize its allocation cost. 

More precisely, consider a utility company which can allocate energy to $K\in \mathbb{N}$ different substations. Each substation can distribute the received energy from the utility company to a disjoint subset of prosumers in the grid. We distinguish each substation by the set of prosumers that it serves and denote them by $\mathcal{B}_1,\mathcal{B}_2,\ldots,\mathcal{B}_{K}$. In fact, substations are the main link between the prosumers and the utility company as illustrated in Figure \ref{fig:smart-grid}. At the beginning of each time $t=1,2,\ldots$, the utility company decides on the amount of energy to generate and how to distribute it among different substations, denoted by $e_1(t),\ldots,e_K(t)$. The action set of the utility company at any time $t$ will be:
\begin{align}\label{eq:action-set-utility-company}
\mathcal{E}\!=\!\Big\{{\rm \boldsymbol{e}}\!=\!\big(e_1,\ldots,e_K\big)\!: \ e_i\ge 0, \forall i, \ \sum_{i=1}^{K}e_i\!\leq\! E_{\max}\Big\},
\end{align}
where $E_{\max}$ is the maximum energy generation capability of the utility company. 

At time $t$, each prosumer demands extra energy $D_i(t)$ at the price $p_i(\boldsymbol{D}(t))$ from its closest substation. If such a trading results in a shortage of energy at some substation $\mathcal{B}_{\ell}$, the utility company incurs a cost proportional to the extra energy that it has to produce in order to fully satisfy the prosumers' needs in that substation. In other words, the utility company benefits from selling energy to the grid, but suffers a cost due to its prediction error in energy generating at time $t$ captured by $(\sum_{j\in \mathcal{B}_{\ell}}D_j(t)-e_{\ell}(t))^2$. Given an allocation profile ${\rm \boldsymbol{e}}(t)\in \mathcal{E}$ and some arbitrary but fixed pre-selected energy pricing function $\boldsymbol{p}(\cdot)$, we define the instantaneous cost of the utility company at time $t$ to be
\begin{align}\label{eq:company-cost}
C({\rm \boldsymbol{D}}(t),{\rm \boldsymbol{e}}(t)):&=\beta\sum_{\ell=1}^{K}e_{\ell}(t)-\sum_{i=1}^{N}D_i(t)p_i(\boldsymbol{D}(t))\cr 
&\qquad+\gamma\sum_{\ell=1}^{K}\Big(\sum_{j\in \mathcal{B}_{\ell}}D_j(t)-e_{\ell}(t)\Big)^{2}\!\!,
\end{align}
where $\beta$ and $\gamma$ are constant denoting the actual unit price of energy generation and regeneration (due to prediction error) for the utility company. The first term in \eqref{eq:company-cost} is the generation cost for the utility company based on its prediction, the second term is its income due to selling electricity to the market with determined price function $\boldsymbol{p}(\cdot)$, and the last term corresponds to the extra cost that the utility company incurs due to its false prediction in energy allocation to different substations. 

\begin{figure}[t!]
\vspace{-2cm}
\begin{center}
\hspace{0.5cm}\includegraphics[totalheight=.25\textheight,
width=.37\textwidth,viewport=50 0 750 700]{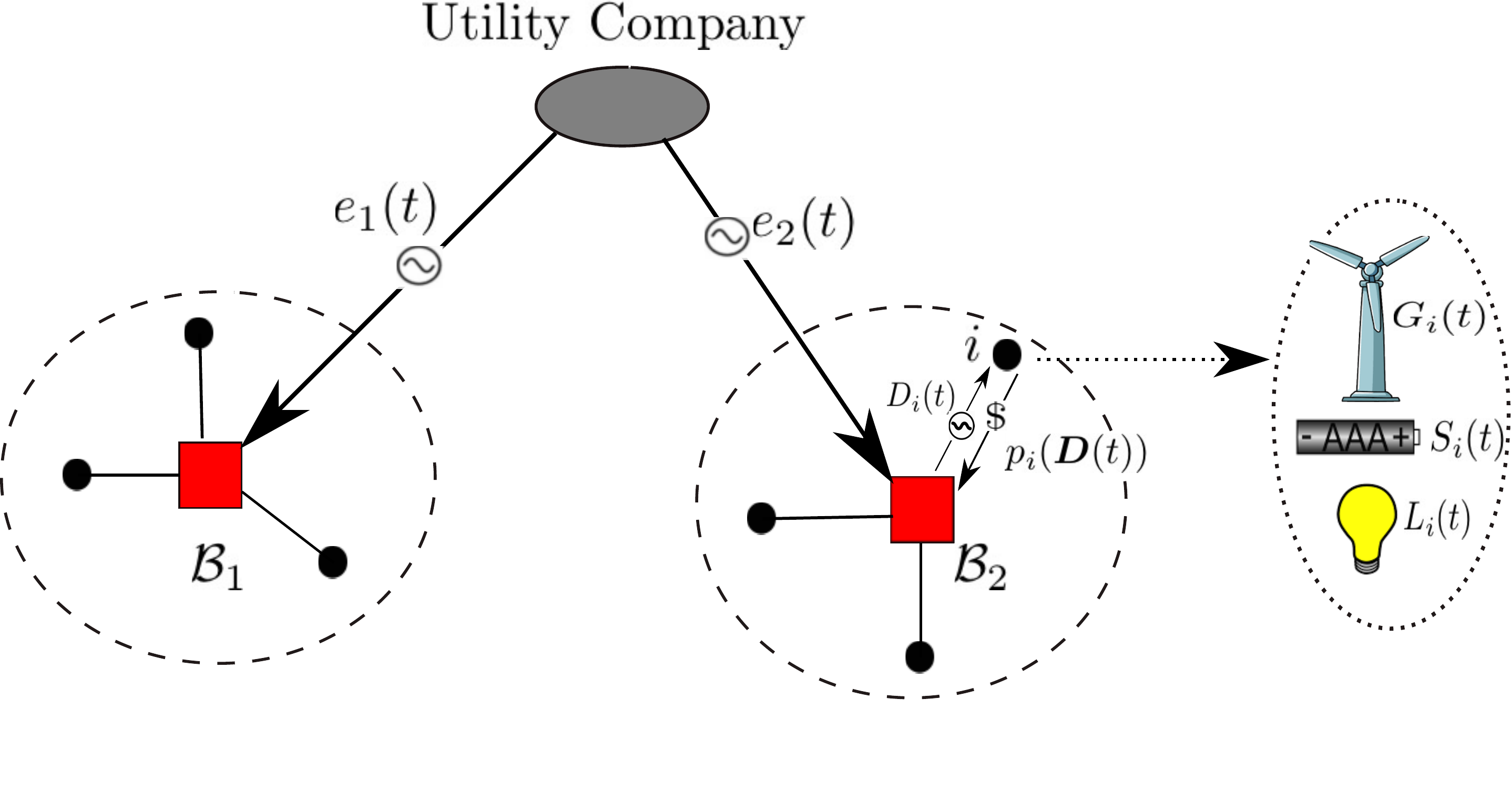} \hspace{0.4in}
\end{center}\vspace{-0.75cm}
\caption{Illustration of the smart grid with $2$ substations and $6$ prosumers.}\vspace{-0.3cm}
\label{fig:smart-grid}
\end{figure} 

As mentioned earlier, the utility company must generate enough energy and distribute it among different substations in order to satisfy prosumers' demands. If the utility company were aware of the sequence of energy demands by prosumers a priori, then it could easily solve an open loop optimization problem to find its fixed optimal allocation rule in order to minimize its overall cost. However, one important challenge here is that since the prosumers' demands heavily depend on the outcome of the stochastic game among them, it is not clear how much energy will be requested by different prosumers at different stages of the game. Therefore, one can define the regret of the utility company to be the difference between its minimum costs when the utility company is aware of the demands a priori and when it is not. This provides a reasonable measure on how well the utility company can predict the demand market, and hence, respond to it properly. Therefore, the utility company aims to learn the optimal allocation rules in order to minimize its overall regret given by
\begin{align}\nonumber
R_T:=\sum_{t=1}^{T}C({\rm \boldsymbol{D}}(t),{\rm \boldsymbol{e}}(t))-\max_{\boldsymbol{x}\in \mathcal{E}}\sum_{t=1}^{T}C({\rm \boldsymbol{D}}(t),\boldsymbol{x}).
\end{align}
which is a measure frequently used in the computer science literature for evaluating the performance of a learning strategy with respect to an uncertain environment. Here, the utility company's goal is to select a strategy that minimizes its regret. This means that it dynamically allocates energy units \textit{without} knowing the actual demands with a similar cost as if it \textit{knows} the entire demand market ahead of time and wants to find a fixed optimal allocation.


\section{Stationary NE Among Prosumers}\label{sec:NE}

Next, we analyze the interaction among the prosumers formulated as a stochastic game in Section \ref{sec:stochastic-game}. Often in smart grid there are many prosumers who are simultaneously maximizing their own payoffs. Therefore, a global optimal solution which maximizes all prosumers' payoffs may not exist. One suitable concept to solve this game, is that of a Nash equilibrium (NE) which, for the studied stochastic game can be formally defined as follows:
\begin{definition}
\normalfont A policy profile $\boldsymbol{\Psi}^*=(\Psi^*_1,\ldots,\Psi^*_N)$ is said to constitute a closed-loop \emph{Nash equilibrium} (NE) if for all $i\in \mathcal{N}$, and any policy $\Psi_i$ of a given player $i$ with payoff function $V_i(\cdot)$ given in \eqref{eq:discounted-payoff}, we have $V_i(\Psi^*_i,\boldsymbol{\Psi}_{-i}^*)\ge V_i(\Psi_i,\boldsymbol{\Psi}_{-i}^*)$.
\end{definition}
In other words, a policy profile constitutes an NE of the stochastic game among prosumers if no prosumer can unilaterally improve its payoff by changing its policy. Therefore, at a given NE, each prosumer can only look at its own energy storage level and then decide on how much energy to consume and to purchase from the utility company in order to maximize its own payoff. Next, we state the following assumption which essentially requires that the generated energy process of different prosumers be independent. 


\begin{assumption}\label{ass:stationary}
\normalfont  We assume that the generated energy processes of prosumers form i.i.d random processes such that $\lambda_i\!:\!=\!\min\limits_{|k|\leq S^{\max}_i}\mathbb{P}\{G_i\!=\!k\}\!>\!0, \forall i$, where $G_i$ denotes the generated energy distribution of prosumer $i$ such that $G_i(t)\sim G_i, \forall t$. Moreover, we assume that prosumers have limited computational capabilities such that they cannot estimate each others' policies by only looking at their own payoffs.
\end{assumption} 
To justify the above assumption, extensive statistical experiments in weather forecasting have shown that wind speed often matches the Weibull or Rayleigh distributions \cite{carta2009review}. Note that Assumption \ref{ass:stationary} allows different kinds of random generated energy $G_i$ for different prosumers as long as they have reasonably large support to ensure that $\lambda_i>0$ (which is the case for Weibull or Rayleigh distribution). Moreover, as discussed in \cite{hjorth2000snow}, it is generally assumed that the weather system with a time separation of every 3 to 4 days can be assumed to be independent. In particular, if the prosumers are located in relatively far distances from each other, then they are most likely subject to independent environmental conditions, i.e., independent generated energy. It is worth mentioning that the independencies of generated energy processes of different prosumers can be relaxed to the cases in which there is a correlation among them. However, in that case we need a stronger assumption which guarantees that for each prosumer $i$ and for any stationary policy $\Psi_i$ of that prosumer, the induced Markov chain of player $i$ over its states admits a unique stationary distribution \cite{altman2008constrained}.

We now state the following Lemma which says that any stationary policy followed by a prosumer induces a stationary distribution over its set of states, and that the convergence to such a stationary distribution is geometrically fast.       

\begin{lemma}\label{lemm:stationary}
\normalfont Under Assumption \ref{ass:stationary}, for any stationary policy $\Psi_i$ of prosumer $i$,  there exists a unique probability vector $\pi^{\Psi_i}$ over the states $\mathcal{S}_i$ such that 
\begin{align}\nonumber
\max_{s\in \mathcal{S}_i}\left|\mathbb{P}^{\Psi_i}\{S_i(t)=s\}-\pi^{\Psi_i}(s)\right|\leq (1-\lambda_i)^t,
\end{align}  
where $\mathbb{P}^{\Psi_i}\{\cdot\}$ denotes the probability measure induced over the state space of prosumer $i$ when following policy $\Psi_i$.  
\end{lemma}
\begin{proof}
A proof can be found in Appendix \ref{apx:lemma-stationary}.
\end{proof}

Next, in order to establish the existence of NE among prosumers, our first step is to characterize the best-response set of each prosumer with respect to others. For this purpose, we use a similar technique as in \cite{altman1999constrained} to characterize the optimal policy of each prosumer with respect to others using a linear program with occupation measures as its variables. Denote an arbitrary but fixed stationary policy for all prosumers other than $i$ by $\boldsymbol{\Psi}_{-i}$. Following any policy $\theta_i$ (not necessarily stationary) by prosumer $i$ induces a probability measure $\mathbb{P}^{\theta_i}\{\cdot\}$ on the trajectories of its joint state-actions $(S_i(t),A_i(t))$. On the other hand, since the dynamics of state-actions of each prosumer is totally determined by its own policy and its own generated energy process which by Assumption \ref{ass:stationary} is independent of others, the induced probabilities over the joint action-space of different prosumers are independent. Thus
\begin{align}\nonumber
&\mathbb{P}^{(\theta_i,\boldsymbol{\Psi}_{-i})}\{\boldsymbol{S}(t)={\rm \boldsymbol{s}}, \boldsymbol{A}(t)={\rm \boldsymbol{a}}\}\cr 
&\!=\!\mathbb{P}^{\theta_i}\{S_i(t)\!=\!s_i,A_i(t)\!=\!a_i\!\}\prod_{j\neq i}\!\mathbb{P}^{\Psi_j}\{S_j(t)\!=\!s_j, A_j(t)\!=\!a_j\!\}. 
\end{align}
Using \eqref{eq:discounted-payoff}, the payoff of the $i$th prosumer can be written as 
\begin{small}\begin{align}\label{eq:prospect-evaluation}
V_i(\theta_i,\boldsymbol{\Psi}_{-i})&=\lim_{T\to \infty}\frac{1}{T}\sum_{t=1}^{T}\sum_{({\rm \boldsymbol{s}},{\rm \boldsymbol{a}})}\mathbb{P}^{\theta_{i}}\{S_{i}(t)=s_{i}, A_{i}(t)=a_{i}\}\cr 
&\times w_i\Big(\prod_{j\neq i}\mathbb{P}^{\Psi_j}\{S_j(t)=s_j, A_j(t)=a_j\}\Big)v_i\big(U_i(\boldsymbol{a})\big),
\end{align}\end{small}where $w_i(\cdot)$ and $v_i(\cdot)$ are the weight and valuation functions of a prosumer $i$ as described in Definition \ref{def:prospect}. In \eqref{eq:prospect-evaluation} we have assumed that a prosumer has a subjective evaluation only of the other players' probabilities. This is because usually each prosumer $i$ is aware of its own prospect due to the fact that $\mathbb{P}^{\theta_i}\{S_i(t)=s_i,A_i(t)=a_i\}$ is fully determined based on its own policy which is known to prosumer $i$ (and hence it does not weight that probability), but it has uncertainty on others' policies (consequently their outcome probabilities). Then,
\begin{align}\label{eq:w-limit}
&\lim_{t\to \infty}w_i\Big(\prod_{j\neq i}\mathbb{P}^{\Psi_j}\{S_j(t)=s_j, A_j(t)=a_j\}\Big),\cr 
&=w_i\Big(\lim_{t\to \infty}\prod_{j\neq i}\mathbb{P}^{\Psi_j}\{A_j(t)=a_j|S_j(t)=s_j\}\mathbb{P}\{S_j(t)=s_j\}\Big),\cr 
&=w_i\Big(\lim_{t\to \infty}\prod_{j\neq i}\Psi_j(a_j|s_j)\mathbb{P}\{S_j(t)=s_j\}\Big),\cr 
&=w_i\Big(\prod_{j\neq i}\Psi_j(a_j|s_j)\pi^{\Psi_j}(s_j)\Big), 
\end{align}where the first equality follows from the continuity of $w_i(\cdot)$, the second equality holds because $\Psi_j$ is a stationary policy which does not depend on time, and the last equality is due to Lemma \ref{lemm:stationary}. Substituting \eqref{eq:w-limit} into \eqref{eq:prospect-evaluation} we get 

\vspace{-0.3cm}
\begin{small}
\begin{align}\label{eq:V_i}
&V_i(\theta_i,\boldsymbol{\Psi}_{-i})=\lim_{T\to \infty}\frac{1}{T}\sum_{t=1}^{T}\sum_{({\rm \boldsymbol{s}},{\rm \boldsymbol{a}})}\mathbb{P}^{\theta_{i}}\{S_{i}(t)=s_{i}, A_{i}(t)=a_{i}\}\cr 
&\qquad\qquad\qquad\times w_i\Big(\prod_{j\neq i}\Psi_j(a_j|s_j)\pi^{\Psi_j}(s_j)\Big)v_i\big(U_i(\boldsymbol{a})\big),\cr
&=\lim_{T\to \infty}\frac{1}{T}\sum_{t=1}^{T}\sum_{(s_i,a_i)}\Big[\mathbb{P}^{\theta_i}\{S_{i}(t)=s_{i}, A_{i}(t)=a_{i}\}\cr 
&\qquad\qquad\qquad\times\!\!\!\!\!\sum_{({\rm \boldsymbol{s}}_{-i},{\rm \boldsymbol{a}}_{-i})}\!\!\!\!\!w_i\big(\prod_{j\neq i}\Psi_j(a_j|s_j)\pi^{\Psi_j}(s_j)\big)v_i\big(U_i(\boldsymbol{a})\big)\Big],\cr 
&=\lim_{T\to \infty}\frac{1}{T}\sum_{t=1}^{T}\sum_{(s_i,a_i)}\mathbb{P}^{\theta_i}\{S_{i}(t)=s_{i}, A_{i}(t)=a_{i}\}\times K_i(s_i,a_i).
\end{align}\end{small}Here, $K_i(s_i,a_i)\!:=\!\!\!\!\!\!\sum\limits_{({\rm \boldsymbol{s}}_{-i},{\rm \boldsymbol{a}}_{-i})}\!\!\!\!\!\!w_i\big(\prod\limits_{j\neq i}\Psi_j(a_j|s_j)\pi^{\Psi_j}(s_j)\big)v_i\big(U_i(\boldsymbol{a})\big)$. As seen from \eqref{eq:V_i}, for fixed stationary policies $\boldsymbol{\Psi}_{-i}$ of other players, the optimal policy for player $i$ can be obtained by solving a Markov Decision Process (MDP) whose objective function $V_i(\theta_i,\boldsymbol{\Psi}_{-i})$ depends only on the marginal distribution of the $i$th player policy $\theta_i$ induced on its joint state-action pair $(S_i(t),A_i(t))$. As has been shown in \cite[Theorem 4.1]{altman1999constrained}, for any such MDP, the set of stationary policies is complete and dominant, meaning that without any loss of generality prosumer $i$ can obtain its optimal policy among stationary policies. Therefore, letting $\theta_i:=\Psi_i$ be the optimal stationary policy of prosumer $i$ with respect to $\boldsymbol{\Psi}_{-i}$, one can easily see that $\lim\limits_{T\to \infty}\frac{1}{T}\sum\limits_{t=1}^{T}\mathbb{P}^{\Psi_{i}}\{S_{i}(t)\!=\!s_{i}, A_{i}(t)\!=\!a_{i}\}$ exists, and equals $\Psi_i(a_i|s_i)\pi^{\Psi_i}(s_i)$. Now similar as to \cite{altman1999constrained}, define the occupation measures $\{\rho(s_i,a_i), \forall (s_i,a_i)\in \mathcal{S}_i\times \mathcal{A}_i \}$ to be 
\begin{align}\label{eq:occupation-def}
\rho(s_i,a_i):=\lim_{T\to \infty}\frac{1}{T}\sum_{t=1}^{T}\mathbb{P}^{\Psi_{i}}\{S_{i}(t)=s_{i}, A_{i}(t)=a_{i}\}.
\end{align}
 
An occupation measure corresponding to a policy can be considered as a probability measure over the set of state-action which determines the proportion of time that the policy spends over each joint action-state. An important property of such formulation is that the average payoff corresponding to that policy can be expressed as the expectation of instantaneous payoffs with respect to occupation measure. Now using  \cite[Theorem 4.2 and Eq. (4.3)]{altman1999constrained}, we can write the policy optimization problem for prosumer $i$ given the stationary policy $\boldsymbol{\Psi}_{-i}$ of its opponents as the solution to the following linear program:
\begin{align}\label{eq:LP}
&\max \ \sum_{(s_i,l_i)}\rho_i(s_i,a_i)K_i(s_i,a_i),\cr 
&s.t. \ \ \sum_{(s_i,a_i)}\!\!\rho_i(s_i,a_i)(\mathbb{I}_{\{x=s_i\}}-W^i_{xa_is_i})=0, \ \ \forall x\in \mathcal{S}_i, \cr 
&\ \ \ \ \ \ \ \rho_i(s_i,a_i)\ge 0,
\end{align}where $W^{(i)}_{xa_is_i}\!=\!\mathbb{P}\{S_i(t+1)\!=\!x|S_i(t)\!=\!s_i,A_i(t)\!=\!a_i\}$ denotes the probability that the state of prosumer $i$ will change from $s_i$ to $x$, by talking action $a_i$. Note that this quantity is totally determined once the distribution of generated energy $G_i$ is known. Denoting the optimal solution of the above maximization problem by ${\rm \boldsymbol{\rho}}_i:=\{\rho_i(s_i,a_i), \forall (s_i,a_i)\in \mathcal{S}_i\times \mathcal{A}_i \}$ (for notational simplicity), one can construct a stationary policy $\Psi_i$ for prosumer $i$ as:
\begin{align}\label{eq:occupation-to-policy}
\Psi_i(a_i|s_i)=\frac{\rho_i(s_i,a_i)}{\sum_{a_j\in \mathcal{A}_i}\rho_i(s_i,a_j)}.
\end{align}
As shown in \cite[Theorem 4.1]{altman1999constrained}, such a stationary policy induces the same occupation measure as $\rho_i$, which implies that the prospect payoff of the prosumer $i$ by following policy $\Psi_i$ is the same as the optimal solution of the linear program \eqref{eq:LP}, i.e., $\Psi_i$ is an optimal stationary policy for prosumer $i$. Using this characterization, we have the following theorem:  

\begin{theorem}\label{thm:existence}
\normalfont There exists a stationary NE policy for the stochastic game among prosumers with prospect payoffs \eqref{eq:discounted-payoff}. 
\end{theorem}
\begin{proof}
See Appendix \ref{apx:existence}. 
\end{proof}

The proof of Theorem \ref{thm:existence} follows by reformulating the original stochastic game among prosumers as a ``\textit{virtual}" game in normal form using the LP characterization given in \eqref{eq:LP}, and showing that the virtual game admits a pure-strategy NE, which in turn implies the existence of a stationary NE in the original stochastic game. The virtual game is composed of $N$ players (one for each prosumer) such that the action set of player $i$ equals to the set of feasible occupation measures for prosumer $i$ in the original stochastic game given by 

\vspace{-0.3cm}
\begin{small}
\begin{align}\label{eq:action-occupation-set}
\mathcal{M}_i\!:=\!\Big\{&\rho_i\ge 0: \!\!\sum_{(s_i,a_i)}\!\!\rho_i(s_i,a_i)(\mathbb{I}_{\{x=s_i\}}\!-\! W^{(i)}_{xa_is_i})\!=0, \forall x\in \mathcal{S}_i\Big\}.
\end{align}\end{small}In particular, an action for player $i$ in the virtual game means choosing an occupation measure for prosumer $i$. For more details about the virtual game, we refer to Appendix \ref{apx:existence}.

We note that, in general, using PT rather than EUT can eliminate the possibility of existence of an NE (see e.g., \cite{metzger2010equilibria}). However, one of the advantages of our stochastic game formulation is its rich structure which allows us to conclude the same existence results even under PT. Here, we should mention that, while dealing with the analytical study of games under PT is more challenging, this is not the main reason why we incorporated the role of PT into our model. In fact, PT captures real-world subjective behavior of humans \cite{kahneman1979prospect} with substantial implications in real life events, as we will see under our model in Section \ref{sec:simulation}.

Theorem \ref{thm:existence} yields two key insights: 1) It guarantees existence of simple policies for the prosumers such that simultaneously satisfy all the prosumers. As a result, the prosumers do not really need to use sophisticated time dependent decision rules in order to maximize their payoffs and simple stationary policies are good enough. 2) It shows that there exists a way to stabilize the electricity market where the prosumers have no further incentive to change their energy consumption or energy demand policies. Consequently, this facilitates the prediction of energy production and distribution for the utility company which, in turn, improves the availability of energy to prosumers, and hence the reliability of the entire grid. Now, an important question is to see whether the prosumers can jointly reach one of such equilibrium points as studied next.

\section{Learning an $\epsilon$-Nash Equilibrium Policy with Limited Information}\label{sec:learning-prosumers}

As discussed in Section \ref{sec:NE}, despite the fact that prosumers cannot observe the entire state of the game or each others' actions, there still exists a stationary NE in the system. In fact, this is one of the main reasons why such a stochastic game is desirable for energy trading in smart grids. This is because, in reality, the prosumers do not really have access to private information of others such as their stored energy or consumption level. Even if such information were available, to avoid extra computations or complete trust on the released data, the prosumers may not want to use them as a baseline in order to select their own policies. Moreover, once the electricity market reaches a stable NE, all the prosumers can benefit in a sense that everyone will be satisfied with its own payoff given that others do not deviate. Therefore, an important question one may ask is how can prosumers achieve one of such NE policies with very limited information about each other, i.e., by only observing their own instantaneous payoffs. This mandates designing a suitable learning algorithm under stochastic setting in which prosumers can infer necessary information for steering their policies toward an NE by only observing their own instantaneous payoffs. In this section, we propose a distributed algorithm that enables the prosumers to jointly achieve a policy profile with almost the same performance as an NE, i.e., an $\epsilon$-NE policy:

\begin{definition}
\normalfont Given $\epsilon>0$, a policy profile $\tilde{\boldsymbol{\Psi}}$ is called an $\epsilon$-NE if no player can unilaterally improve its prospect payoff by more than $\epsilon$, i.e., $V_i(\tilde{\Psi}_i,\tilde{\boldsymbol{\Psi}}_{-i})\ge V_i(\Psi_i,\tilde{\boldsymbol{\Psi}}_{-i})-\epsilon$.
\end{definition}


Before we proceed further, we first state some of the algorithm design challenges. The proposed algorithm must operate under the least amount of \emph{information exchange}. In other words, following a policy, at each time $t=1,2,\ldots$, each prosumer can only observe its own instantaneous prospect payoff at that time, but it has no information about the payoffs received by the others, nor can it observe what policies or actions are taken by others. In particular, a prosumer does not know the structure of its own payoff function and can only receive its value by taking one action at each time $t$. To address this issue, we use some techniques from multi-armed bandit problems \cite{cesa2006prediction}. 


The key idea in our algorithm is that instead of dealing with an infinite horizon stochastic model for the original stochastic game, we devise a learning algorithm for its virtual counterpart. As previously discussed, one can reformulate the stochastic game among prosumers as a normal form game with occupation measures being players' actions. Therefore, our goal is to provide a scheme such that as long as it is followed by prosumers, they repeatedly learn the equilibrium occupation measures. Once such measures are learned, each prosumer can reconstruct its own stationary policy at equilibrium using the rule \eqref{eq:occupation-to-policy}. Here, we note that although the original stochastic game can be played for infinitly many steps, but its equivalent virtual game can be played only once. This is because the virtual game is not a repeated game but a game which compresses all the information (such as payoffs) of the infinite horizon stochastic game into only one shot. Thus, we must extract the required information in our algorithm without playing the virtual game repeatedly and by only playing the original stochastic game. The key observation here is that any time during which we need to evaluate the payoffs received in the virtual game, we do not need to play the original stochastic game until the end. However, playing for a sufficiently long time will give us enough information for the purpose of our algorithm design. This will be formally stated in the following lemma: 

\begin{lemma}\label{lemm:epsilon-value-error}
\normalfont Given $\epsilon>0$, there exists a time period $T(\epsilon)$ such that for any profile of stationary policies $\boldsymbol{\Psi}$ if the prosumers follow $\boldsymbol{\Psi}$ for $T(\epsilon)$ steps without changing policy, then 
\begin{align}\nonumber
|V_i(\boldsymbol{\Psi})-V^{T(\epsilon)}_i(\boldsymbol{\Psi})|<\epsilon,\ \  \forall i=1,\ldots,N,
\end{align}
where $V^{T(\epsilon)}_i(\boldsymbol{\Psi}):=\frac{1}{T(\epsilon)}\sum_{t=1}^{T(\epsilon)}\mathbb{E}^{^{\rm PT}}_{\boldsymbol{\Psi},\boldsymbol{G}}\big[U_i\big(\boldsymbol{A}(t)\big)\big]$ denotes the cumulative payoff of prosumer $i$ up to $T(\epsilon)$ steps.
\end{lemma}
\begin{proof}
See Appendix \ref{apx:epsilon-value-error}.
\end{proof}

Finally, unlike many multi-armed bandit problems with finite action space, the considered virtual game contains a continuum of actions $\mathcal{M}_i$ given by \eqref{eq:action-occupation-set}. In this regard, the optimization problem of finding the best policy for each prosumer can be formulated as a linear program over a bounded polytope $\mathcal{M}_i$. Since each linear program attains its maximum value in at least one of its vertex (extreme) points, to estimate the best response of each player in the virtual game, we only need to estimate the payoffs received by playing the vertex points of $\mathcal{M}_i$. Note that since each of polytopes $\mathcal{M}_i$ is defined using finitely many linear constraints, the total set of vertex points are finite. In fact. without any loss of generality we may assume that all the action polytopes $\mathcal{M}_i, i=1,\ldots,N$, have the same number of vertex points $r\in \mathbb{N}$, otherwise, we take $r$ to be an upper bound on the total number of vertices.    



The proposed distributed learning algorithm is shown in Algorithm \ref{alg:1}. In this algorithm, each prosumer $i$ holds a sufficiently large $T(\epsilon)$ (as in Lemma \ref{lemm:epsilon-value-error})\footnote{Even if $T(\epsilon)$ is not known a priori, the prosumers can use a doubling trick to find the right $T(\epsilon)$ by doubling their period length every time that Algorithm 1 does not terminate after a sufficiently long time, in which case after at most $\log_2(T(\epsilon))$ times restarting Algorithm 1, all prosumers will find a valid period length $T(\epsilon)$.} and computes the set of all the vertex points of its action set in the virtual game denoted by $\mathcal{F}_i$. This set can be fully determined by prosumer $i$ using its own internal Markov chain once the distribution of $G_i$ is known. After that, each prosumer partitions the entire horizon into ``larger" intervals of the form $[(m-1)(Nr+1)T(\epsilon)+1, m(Nr+1)T(\epsilon)), m=1,2,\ldots$, each of which is partitioned into $Nr+1$ smaller sub-intervals $Z_1,\ldots,Z_{Nr+1}$ of length $T(\epsilon)$. At the beginning of each larger interval, each prosumer selects an occupation measure $\rho_i(m-1)$ from $\mathcal{M}_i$ uniformly at random and constructs its corresponding stationary policy $\Psi_i(m-1)$ using \eqref{eq:occupation-to-policy}. After that during the smaller sub-intervals each prosumer evaluates the optimality of its sampled policy with respect to others by playing its vertex policies within certain sub-intervals, as has been illustrated for the case of 3 prosumers in Figure \ref{fig:alg1_ex}.

\begin{algorithm}[t!]
\caption{Strategy for Prosumer $i$}
\label{alg:1}
\begin{algorithmic}[0]
\State {\bfseries Parameters} Accuracy parameter $\epsilon$, number of players $n$, period length $T(\epsilon)$.
\State {\bfseries Initialize:} Let $\mathcal{F}_i=\{\rho_{i_1},...,\rho_{i_{r}}\}$ be the vertex points of action set $\mathcal{M}_i$, and $\{\Psi_{i_1},...,\Psi_{i_{r}}\}$ be the corresponding vertex (extreme) policies obtained by \eqref{eq:occupation-to-policy}. Choose an action $\rho_i(0)$ uniformly at random and independently from polytope $\mathcal{M}_i$ and compute $\Psi_i(0)$ using \eqref{eq:occupation-to-policy}.   
\For{$t = 1,2,...$}.
\State If $t\in [(m-1)(Nr+1)T(\epsilon)+1, m(Nr+1)T(\epsilon))$ for some integer $m\ge 0$, partition this interval into $Nr+1$ sub-intervals $Z_1,...,Z_{Nr+1}$, each of length $T(\epsilon)$. Follow policies $\Psi_{i_1},...,\Psi_{i_{r}}$ in the sub-intervals $Z_{(i-1)r+1},\ldots,Z_{ir}$, respectively, and policy $\Psi_{i}(m-1)$ in the remaining sub-intervals. Denote the policy of all other prosumers in the $j$th sub-interval by $\boldsymbol{\Psi}_{-i}(Z_j)$. Let
\begin{itemize}
\item $\hat{V}_{i_k}^{(m)}=V^{T(\epsilon)}_i(\boldsymbol{\Psi}_{-i}(Z_{(i-1)r+k}),\Psi_{i_k})$
\item $\hat{V}_i^{(m)}=V^{T(\epsilon)}_i(\boldsymbol{\Psi}_{-i}(Z_{Nr+1}),\Psi_{i}(m-1))$
\end{itemize}
\State If $t=m(Nr+1)T(\epsilon)$ for an integer $m\ge 1$, then let 
\begin{align}\nonumber
q^{(m)}_i = \begin{cases} 1 & \mbox{if} \ \ \hat{V}_i^{(m)}>\max\limits_{k=1,\ldots,r}\hat{V}_{i_k}^{(m)}-\epsilon, \\ 
0 & \mbox{otherwise}; \end{cases} 
\end{align}
\State If $t=m(Nr+1)T(\epsilon)$ and all players have played 1, then keep playing $\Psi_i(m-1)$ forever, otherwise, choose $\rho_i(m)$ randomly according to the uniform distribution over $\mathcal{M}_i$.  
\EndFor
\end{algorithmic}
\end{algorithm}

\setlength{\textfloatsep}{10pt}

More specifically, the goodness of the selected sampled policy $\Psi_i(m-1)$ at the beginning of the $m$th larger interval is examined over its smaller sub-intervals. For this purpose, each prosumer switches its policy only at the beginning of sub-intervals. Since each sub-interval is long enough, sticking to a fixed policy in that sub-interval guarantees that a prosumer can estimate its payoff for that fixed policy up to a small additive error (Lemma \ref{lemm:epsilon-value-error}). Moreover, prosumer $i$ plays the randomly selected policy $\Psi_i(m-1)$ in all the sub-intervals, except $r$ of them, namely, \textit{sampling} sub-intervals, in which it plays policies $\{\Psi_{i_1},...,\Psi_{i_{r}}\}$. Note that each prosumer $i$ can find its vertex policies $\Psi_{i_1},\ldots,\Psi_{i_r}$ using $\mathcal{M}_i$ which is fully known to it at the beginning of Algorithm \ref{alg:1} and do not change. Since the sampling intervals of different prosumers do not overlap (see, Figure \ref{fig:alg1_ex}), this allows each prosumer to get an estimate within an small error of its best response policy. Comparing the estimated payoffs received by following $\Psi_i(m-1)$ with those of following $\{\Psi_{i_1},...,\Psi_{i_{r}}\}$, prosumer $i$ can decide whether it is playing its best policy or not (up to a small additive term $\epsilon$). Finally, at the end of each long period, i.e., $t=m(Nr+1)T(\epsilon))$, if all the prosumers are happy with their payoffs, i.e. $q^{(m)}_i=1, \forall i$, this means that an $\epsilon$-NE is obtained and the prosumers do not need to change their policies. Otherwise, they randomly choose new policies to explore more.

Here, $q^{(m)}_i$ can be viewed as a signaling bit which is a binary variable controlled by prosumer $i$ whose value is updated at the end of each longer interval $m=1,2,\ldots$. The value of this bit is always set to zero by player $i$ except when $\hat{V}_i^{(m)}>\max\limits_{k=1,\ldots,r}\hat{V}_{i_k}^{(m)}-\epsilon$, which means that prosumer $i$'s payoff by following policy $\Psi_i(m-1)$ during the $m$th longer interval is at least as good as that for all of its vertex policies minus $\epsilon$. The signaling bit $q^{(m)}_i$ can be thought of as a satisfaction voting survey conducted periodically among prosumers by the utility company in order to stabilize the market faster. At the end of each longer interval $m=1,2,\ldots$, each prosumer $i$ sends its signaling bit to the utility company ($q^{(m)}_i=0$ means that prosumer $i$ is not satisfied with its current policy $\Psi_i(m-1)$, while $q^{(m)}_i=1$ means that it is satisfied). Then, the utility company responds back to all the prosumers by sending them a bit $q^{(m)}=\prod_{i=1}^{N}q^{(m)}_i$. If $q^{(m)}=1$, meaning that all
players' signaling bits are equal to 1 in the $m$th period, then players stop searching and stick to their current policies $\boldsymbol{\Psi}(m-1)$. Otherwise, players will continue exploring other policies by sampling from their occupation measure sets in the next period. The convergence of Algorithm \ref{alg:1} can be shown formally in the following theorem: 
\begin{theorem}\label{thm-alg-convergence}
\normalfont If all the prosumers choose their policies based on Algorithm 1, then almost surely they will converge to an $\epsilon$-NE policy which will be played forever.  
\end{theorem} 
\begin{proof}
See Appendix \ref{apx:alg-convergence}.
\end{proof}
Finally, we should mention that one of the advantages of our proposed approach here was to introduce a distributed learning algorithm that converges to an $\epsilon$-NE policy under PT, although finding closed form solutions for NE policies under PT seems to be more complicated than under EUT due to the extra nonlinearities in the problem.

\begin{figure}[t!]
\vspace{-5.5cm}
\begin{center}
\includegraphics[totalheight=.31\textheight,
width=.44\textwidth,viewport=20 0 920 900]{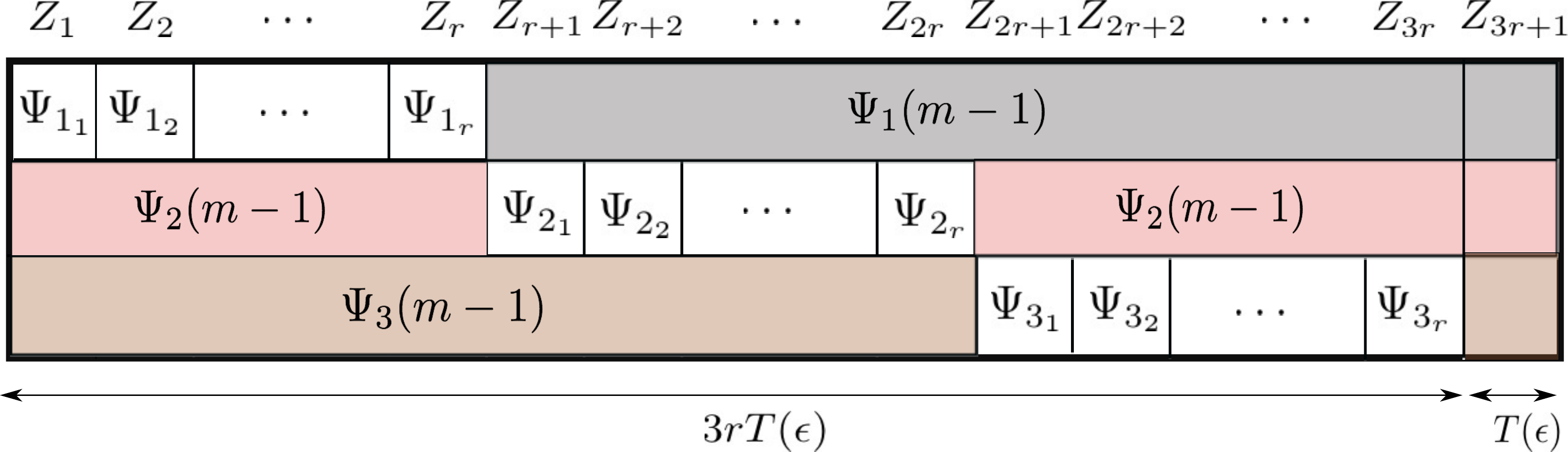} \hspace{0.4in}
\end{center}\vspace{-0.3cm}
\caption{Illustration of sampling intervals in Algorithm \ref{alg:1} for the case of three prosumers ($N=3$) in the larger interval $t\in [(m-1)(3r+1)T(\epsilon)+1, m(3r+1)T(\epsilon))$. Based on this structure, prosumer 1 can compute $\hat{V}_{1_1}^{(m)},\ldots,\hat{V}_{1_r}^{(m)}$ using the policies in the first $r$ columns. Similarly, prosumers 2 and 3 can compute $\hat{V}_{2_1}^{(m)},\ldots,\hat{V}_{2_r}^{(m)}$, and $\hat{V}_{3_1}^{(m)},\ldots,\hat{V}_{3_r}^{(m)}$ within the columns $r+1,\ldots,2r$, and $2r+1,\ldots,3r$, respectively. Finally, each prosumer $i\in[3]$ can compute $\hat{V}_i^{(m)}$ within the last column $3r+1$.}
\label{fig:alg1_ex}
\end{figure}

\section{No Regret Algorithm for the Utility Company}\label{sec:online-convex}

Here, we investigate the optimization problem of the utility company and provide an algorithm for the utility company whose average regret goes to zero as the number of interaction stages increases. Therefore, in the remainder of this section our goal is to provide an adaptive online algorithm which guarantees the average regret of the utility company approaches zero as its number of interactions with the prosumers becomes large. In other words, we show that the utility company can follow a strategy whose average cost is the same as its best fixed strategy if the entire sequence of demands were known. To this end, we consider energy allocation algorithm shown in Algorithm \ref{alg:2} for the utility company which generates and assigns $e_{\ell}(t)$ units of energy at time instant $t$ to the $\ell$th substation.

Although there could be different strategies that the utility company can follow in order to minimize its overall regret, the following theorem asserts that Algorithm \ref{alg:2} provides one such strategies. It is worth noting that in Algorithm \ref{alg:2} the utility company first decides on its energy allocation rule at time $t$ using \eqref{eq:online-alg-2} without actually knowing the demands of prosumers at that time. However, it turns out that the updating rule \eqref{eq:online-alg-2} is sufficient enough to minimize the average regret after sufficiently large number of iterations. As another important feature of Algorithm \ref{alg:2} one can see that it is computationally very cheap and tractable, as it only requires projection of a point on a convex set at each time instant which can be done quite efficiently. 

\begin{algorithm}[t!]
\caption{Utility Company's Allocation Algorithm}
\label{alg:2}
\footnotesize Upon receiving demands $\{\boldsymbol{D}(1),\ldots,\boldsymbol{D}(t-1)\}$ from the prosumers, and allocating energy units $\{\boldsymbol{e}(1),\ldots,\boldsymbol{e}(t-1)\}$ up to time instant $t-1$, the utility company generates and assigns $e_{\ell}(t)$ units of energy at time instant $t$ to the $\ell$th substation given by  
\begin{align}\label{eq:online-alg-2}
\boldsymbol{e}(t)=\prod\nolimits_{\mathcal{E}}\left[\boldsymbol{e}(t\!-\!1)\!-\!\frac{1}{\sqrt{t-1}}\nabla_{\boldsymbol{e}} C\big(\boldsymbol{D}(t\!-\!1),\boldsymbol{e}(t\!-\!1)\big)\right]\!,
\end{align}
where $\Pi_{\mathcal{\mathcal{E}}}[\cdot]$ denotes the projection operator on the set $\mathcal{E}$ given by \eqref{eq:action-set-utility-company}, and $\nabla_{\boldsymbol{e}}C(\cdot)$ denotes the gradient of function $C(\cdot)$ with respect to variable $\boldsymbol{e}$.
\end{algorithm}

\begin{theorem}\label{thm:online-algorithm-utility}
\normalfont The average regret of the utility company by following Algorithm \ref{alg:2} goes to zero as the number of interaction stages increases. In particular, the regret of the utility company after $T$ interactions with the prosumers is bounded above by $O(\sqrt{T})$.  
\end{theorem}   
\begin{proof}
See Appendix \ref{apx:online-algorithm-utility}. 
\end{proof}


As a result of Theorem \ref{thm:online-algorithm-utility}, the utility company can eventually learn its optimal allocation rule in hindsight independent of the outcomes of the game playing among prosumers.

\section{Simulation Results}\label{sec:simulation}

In this section, we evaluate our theortical results using extensive simulations. To provide more concrete results, throughout this section we consider some specific forms for the demand and price functions. We let $D_i(t):=[\tau_i+L_i(t)-S_i(t)]^+$, where $\tau_i\in [0,S_i^{\max}]$ is an internal threshold constant which is only known to prosumer $i$. Note that using this demand function, the action for each prosumer $i$ reduces to only selecting its consumption level. The idea for such a choice of demand function is that we assume each prosumer has a certain threshold $\tau_i$ (e.g., amount of energy that a prosumer usually anticipates to consume in a normal day) such that it buys as much energy from the utility company which, together with the current stored energy $S_i(t)$, satisfies its current consumption $L_i(t)$, and saves an extra $\tau_i$ energy units in the storage for the next time step. Moreover, we adopt the \textit{fairness} pricing function for the utility company (see, e.g., \cite{wang2015load,mohsenian2010optimal}) which charges each prosumer proportional to its demand over the aggregate demand of all others and is given by $p_i(\boldsymbol{D}(t)):=\alpha \frac{D_i(t)}{\sum_{j=1}^{N}D_j(t)}$, where $\alpha$ is a constant which is set by the utility company.   



\begin{figure}[t!]
\vspace{-1.2cm}
\begin{center}
\hspace{0.5cm}\includegraphics[totalheight=.31\textheight,
width=.44\textwidth,viewport=100 0 1000 900]{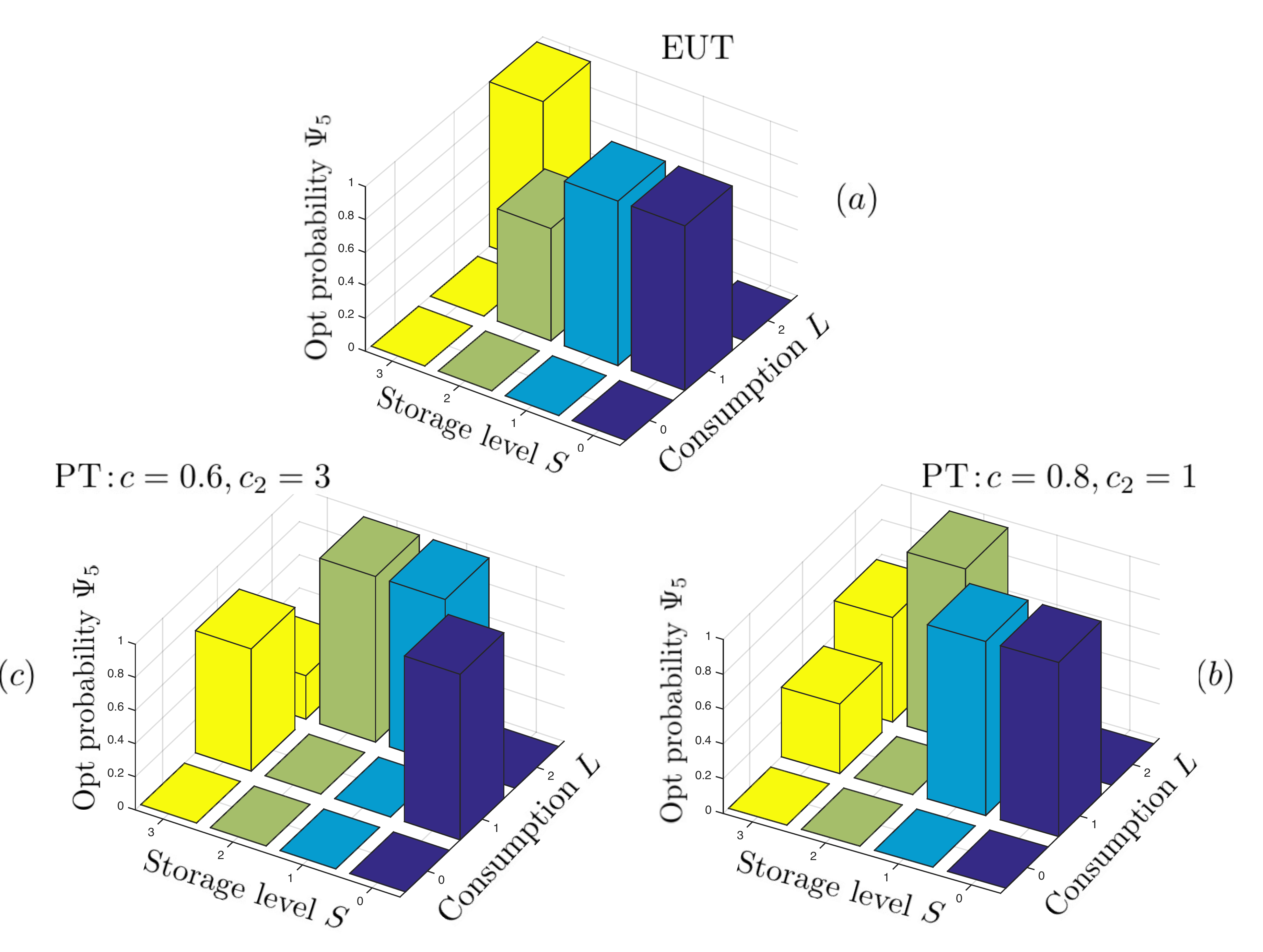} \hspace{0.4in}
\end{center}\vspace{-0.3cm}
\caption{Optimal policy for prosumer $5$ given that all other prosumers follow uniform policies. The top figure depicts this policy when prosumer $5$ maximizes its average expected payoff, while figures \ref{fig:BR-policy}-$(b)$ and \ref{fig:BR-policy}-$(c)$ illustrate its optimal policies under subjective evaluation with two different distortion parameters $c=0.8, c_2=1$ and $c=0.6, c_2=3$.}
\label{fig:BR-policy}
\end{figure}

\subsection{Deviation from Expected Utility}\label{sec:sim-best-response}
First, we will analyze how PT can affect the optimal strategic choices of the prosumers compared with EUT. We consider $N=10$ prosumers with independent Gaussian generated energy distribution with mean and variance vectors $\boldsymbol{\mu}=(0.5, 0.5, 1, 0, 1, 0.7, 0.4, 0.8, 0.3, 1)$ and $\boldsymbol{\sigma}^2=(2, 1, 1, 2, 2, 1, 1, 2, 1, 2)$, respectively. Moreover, we let the internal thresholds of the prosumers be given by a vector $\boldsymbol{\tau}=(1, 0, 1, 2, 1, 0, 0,1,1,2)$. We assume that the storage of each prosumer has four different energy levels in kWh such that $\mathcal{S}_i=\{0,1,2,3\}$, and the consumption level of each prosumer has three different levels $\mathcal{L}_i=\{0,1,2\}$; $0~\mbox{kWh}$ for no consumption, $1~\mbox{kWh}$ for medium consumption, and $2~\mbox{kWh}$ for high consumption. We set the unit price parameter set by the utility company to $\alpha=1$. Furthermore, we assume that the satisfaction function of all prosumers is an increasing concave function such that $f(x)=\log(1+x)$. First, we assume that all prosumers except prosumer $5$ follow a uniform stationary policy in which, independent of their storage levels, they consume $0$, $1$, or $2$ units of energy with equal probabilities of $\frac{1}{3}$ at each stage. We consider two different scenarios. In the first one, we simply assume that prosumer $5$ makes its decisions based on EUT, while in the second case we assume that it selects its policy using PT with weight and value functions given by \eqref{eq:weight-value} with parameters $c\in\{0.6,0.8\}$, $c_1=0.5, c_2\in\{1,3\}$, and $c_3=0.3$.


Figures \ref{fig:BR-policy} illustrates the optimal stationary policy of prosumer $5$ in each of the above scenarios. In this figure the length of the bar in $(s,l)$th coordinate indicates the optimal probability of consuming $l$ units of energy by prosumer $5$ when its stored energy level equals $s$. From Figure \ref{fig:BR-policy}, we can see that there is a considerable difference between optimal policy of prosumer $5$ under EUT (Figure \ref{fig:BR-policy}-$(a)$) and PT (Figure \ref{fig:BR-policy}-$(b),(c)$). In particular, when the parameters in the weight and value functions of PT changes toward more deviation from EUT (i.e., from Figure \ref{fig:BR-policy}-$(b)$ to \ref{fig:BR-policy}-$(c)$), the difference between the optimal policy under EUT and PT will be more substantial. For instance, when prosumer $5$ has higher distortion in its subjective payoff evaluation, it tends to consume more energy within lower ranges of its energy storage. On the other hand, it becomes more conservative for higher energy storage ranges as it becomes more sensitive to future uncertainties.\footnote{Note that the loss in high energy regimes is considerably more than that in low regimes which will be perceived even much more by that prosumer} As an implication, the utility company must allocate more energy units to the substations whose customers act more subjectively, especially when they are experiencing low amount of energy storage.  

\begin{figure}[t!]
\vspace{-2.5cm}
\begin{center}
\hspace{0.5cm}\includegraphics[totalheight=.4\textheight,
width=.45\textwidth,viewport=0 0 600 600]{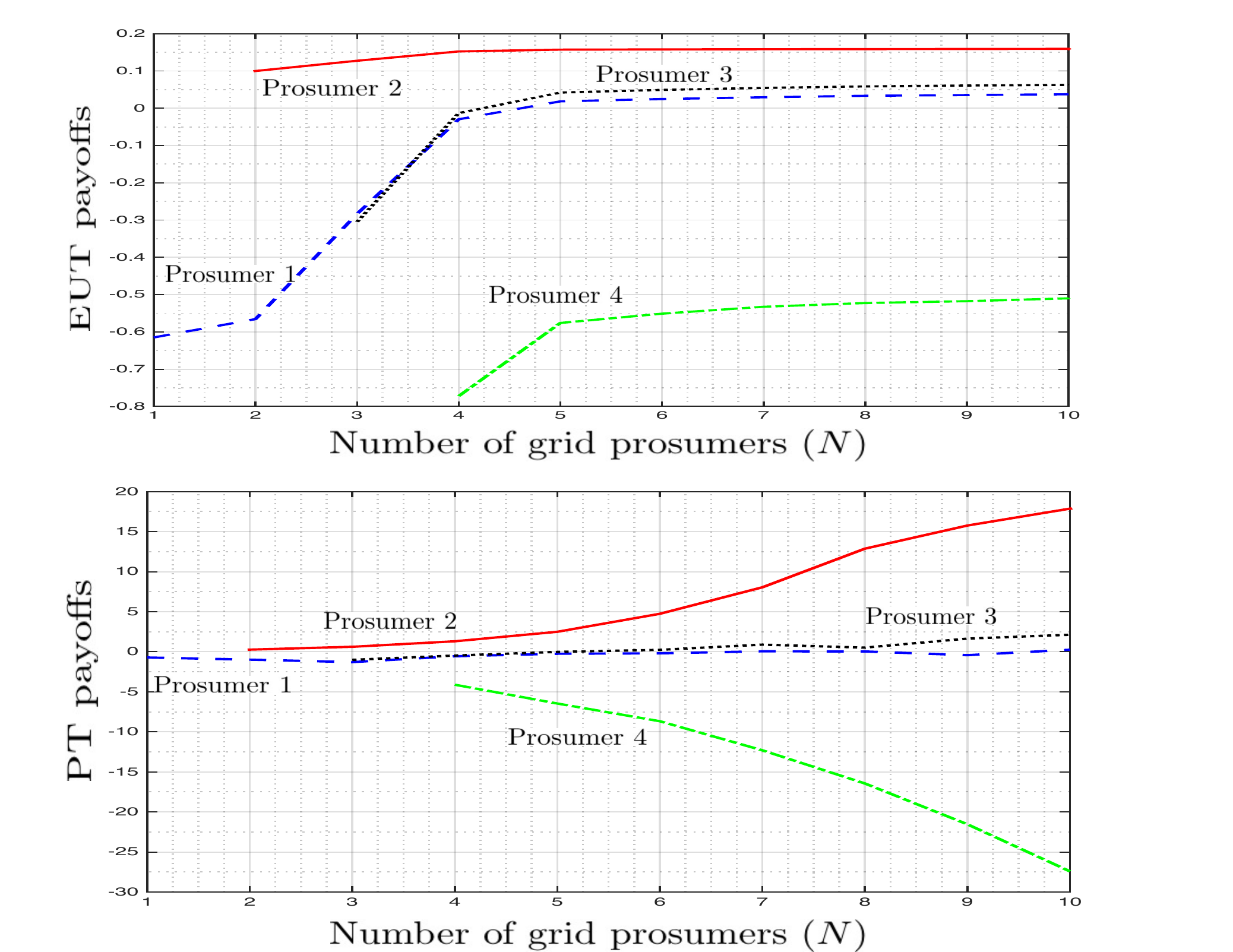} \hspace{0.4in}
\end{center}\vspace{-0.3cm}
\caption{Received payoffs by prosumers $1-4$ under two cases of EUT (left) and PT (right). It is assumed that all prosumers follow uniform policy and they join the grid one by one from prosumer $1$ to prosumer $10$.}
\label{fig:ETPTnumber}
\end{figure}

Figure \ref{fig:ETPTnumber} illustrates the received payoffs by prosumers $1$ to $4$ in terms of the number of players under two cases of PT and EUT. Here, it is assumed that all the players follow a uniform policy, and that they join energy trading sequentially starting from prosumer $1$ up to prosumer $10$. From Figure \ref{fig:ETPTnumber}, we can see that, the prosumers' payoffs under EUT increasingly converge to certain values as more prosumers join the grid. As an example, starting from $-0.6$, the first prosumers' payoff increasingly converges to about $0.05$ as more prosumers join the energy market. One possible reason for such outcome could be the mean field phenomenon which becomes more apparent as more prosumers join the grid. This is because we have assumed that the random generated process of prosumers are independent from each other, such that for large number of prosumers their aggregate behavior converge to certain distribution using central limits theorem. Since a prosumer's payoff strongly depends on the aggregate decisions (e.g., the choice of pricing function), it is reasonable to view a similar type of convergence in prosumers' payoffs. However, this property no longer holds in the case of PT such that prosumers' payoffs follow different pattern without converging to any specific value. This is because in the case of PT, uncertainties will increase as they now stem from both renewable energy and the behavior of the prosumers. As a result, any new prosumer who joins the grid will bring extra uncertainties to the existing ones which ties the prosumers' payoffs in a much more complicated manner.

\subsection{Achieving an $\epsilon$-NE Policy}\label{sim:sec-epsilon-NE}

Next, we illustrate the convergence of Algorithm \ref{alg:1} to an $\epsilon$-NE through a numerical example. Hereinafter, we consider only $3$ different prosumers, each having $3$ storage levels such that $\mathcal{S}=\{0,1,2\}$, and two consumption levels $\mathcal{L}=\{0,1\}$ (each prosumer either decides to consumes energy or not). All other parameters remain similar to Section \ref{sec:sim-best-response} but restricted only to the first three coordinates. For instance, the mean of normal generated energy distribution for these three prosumers are given by $0.5, 0.5$, and $1$, respectively. We set the length of sub-intervals in Algorithm \ref{alg:1} in which the prosumers do not switch their policies to $T(\epsilon)=30$. Finally, we evaluate the policy selection of prosumers under two different scenarios with EUT and PT. The prospect functions of all prosumers are chosen as in Section \ref{sec:sim-best-response} with $c=0.8$.

\begin{figure}[t!]
\vspace{-3cm}
\begin{center}
\hspace{0.5cm}\includegraphics[totalheight=.3\textheight,
width=.5\textwidth,viewport=0 0 700 700]{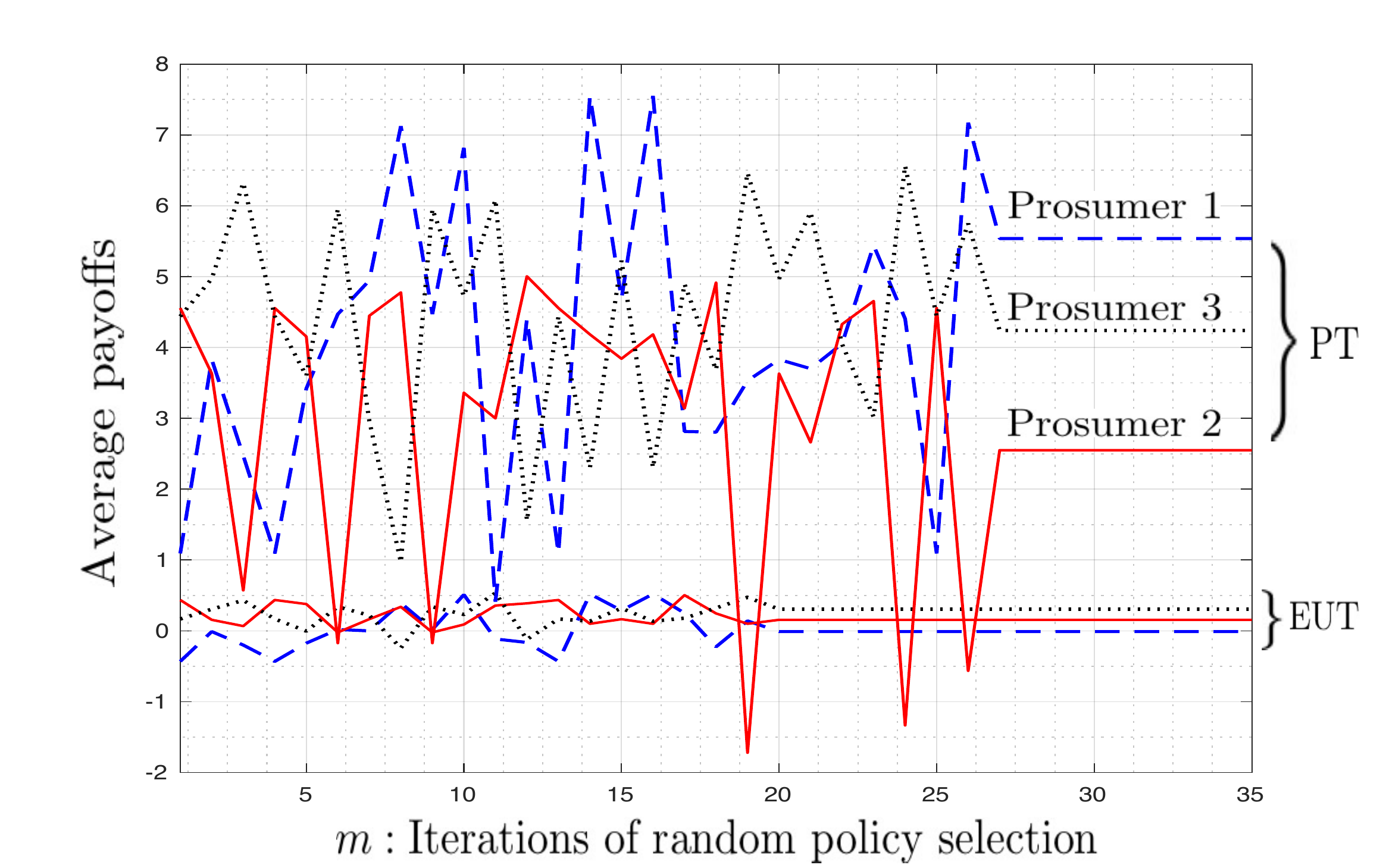} \hspace{0.4in}
\end{center}
\vspace{-0.75cm}
\caption{Illustration of received payoffs by prosumers at the beginning of each long interval in Algorithm \ref{alg:1}. As it can be seen, after $27$ iterations the prosumers with prospect payoffs reach to an $\epsilon$-NE and stick to their policies thereafter. Note that under EUT the prosumers reach to an $\epsilon$-NE even faster.}
\label{fig:epsilon-learning}
\end{figure}

Figure \ref{fig:epsilon-learning} illustrates the trajectories of prosumers' payoffs for different random selection of their policies during executions of Algorithm \ref{alg:1}. Here, the top three curves show the trajectory of prosumers under PT, while the bottom three show that under the EUT. Figure \ref{fig:epsilon-learning} shows that, after each prosumer switches its policy $27$ times, the profile of policies converge to an $\epsilon$-NE, under PT with $\epsilon=0.01$. Then, prosumers will abide by their equilibrium policies and do not change them anymore. Interestingly, one can see that the order of received payoffs by prosumers in the equilibrium under PT is different from that under the EUT. Here, we mention two possible reasons for such a phenomenon: i) The nonlinear relationship between the equilibrium points under PT and EUT due to the weighting and framing effects, and ii) The existence of multiple equilibrium points with different order of payoffs in either EUT or PT. In particular, under the above settings, one can see that the policies converge faster in the EUT case, compared to PT.      


Figure \ref{fig:NE-policies} illustrates the $\epsilon$-NE policies which are achieved at the end of Algorithm \ref{alg:1}. Here, for larger values of storage, the prosumers decide to consume energy more confidently, while for smaller storage values depending on their internal parameters they become more risk averse and randomize their strategies between using and not using energy. An interesting and somehow counter-intuitive behavior in the PT case can be seen in the behavior of prosumer 2 in Figure \ref{fig:NE-policies}. This prosumer has internal parameter $\tau_2=0$, in which case when its storage level is $1$, the expected utility suggests that this prosumer must consume $1$ unit of energy to improve its instantaneous payoff. However, Figure \ref{fig:NE-policies} shows that prosumer $2$ does not consume energy when $S=1$. This is because consuming energy at the current stage will bring the storage level of prosumer $2$ close to zero with high probability, in which case, very likely, this prosumer will suffer lack of energy in the next time step. Since the prosumer tends to overestimate its loss, it decides to postpone its energy usage to the future when it has enough energy in its storage.

\begin{figure}[t!]
\vspace{-3.2cm}
\begin{center}
\hspace{0.5cm}\includegraphics[totalheight=.25\textheight,
width=.4\textwidth,viewport=150 0 1050 900]{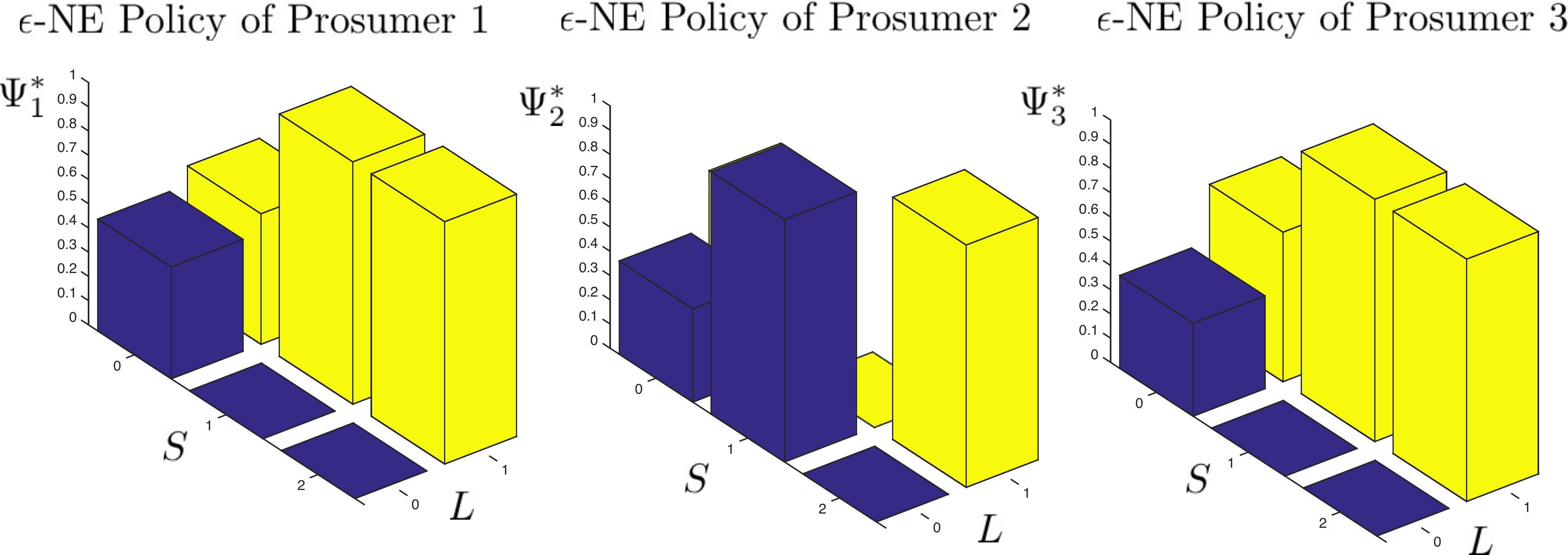} \hspace{0.4in}
\end{center}\vspace{-0.35cm}
\caption{$\epsilon$-NE policies for each of the prosumers with prospect payoffs at the end of Algorithm \ref{alg:1}.}
\label{fig:NE-policies}
\end{figure}

\subsection{Energy Allocation Without Regret}


Finally, we illustrate the performance of the utility allocation Algorithm \ref{alg:2} for the same set of parameters as in Section \ref{sim:sec-epsilon-NE}. Here, we consider three substations (one for each prosumer). Here, we have set the cost of unit energy production and reproduction for the utility company in \eqref{eq:company-cost} to $\beta=1$ and $\gamma=3$, respectively. As it can be seen in Figure \ref{fig:regret_minimization}, as the number of interactions between the utility company and the prosumers increases, the utility company is able to reduce its average regret by properly allocating energy units to the prosumers. In particular, once the demand market has been stabilized at an $\epsilon$-NE, the utility company's regret approaches zero more smoothly because some uncertainty due to prosumers' policy switchings has been ceased. This in turn makes the energy demand market more predictable for the utility company.


\begin{figure}[t!]
\vspace{-5cm}
\begin{center}
\includegraphics[totalheight=.36\textheight,
width=.4\textwidth,viewport= 100 0 1000 900]{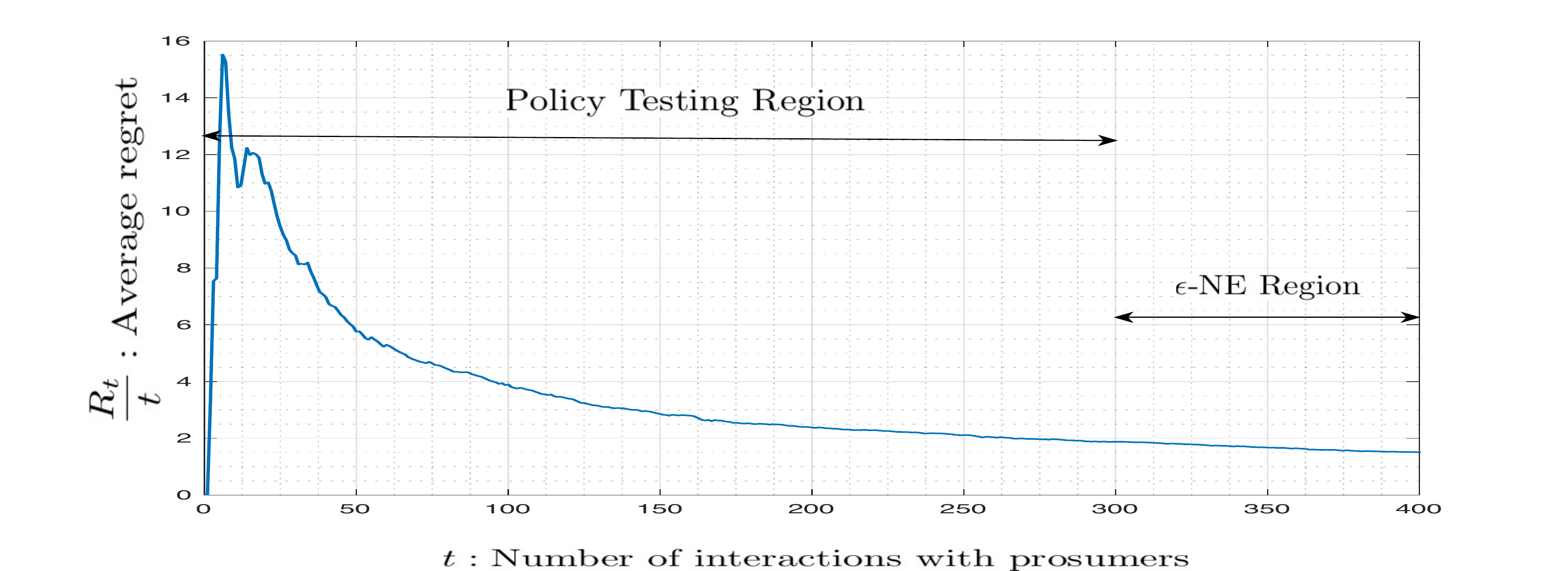} \hspace{0.4in}
\end{center}\vspace{-0.3cm}
\caption{This figure shows the vanishing average regret of the utility company as a result of its interactions with prosumers.}
\label{fig:regret_minimization}
\vspace{0.7cm}
\begin{center}\vspace{0.7cm}
\includegraphics[totalheight=.33\textheight,
width=.4\textwidth,viewport= 30 0 330 300]{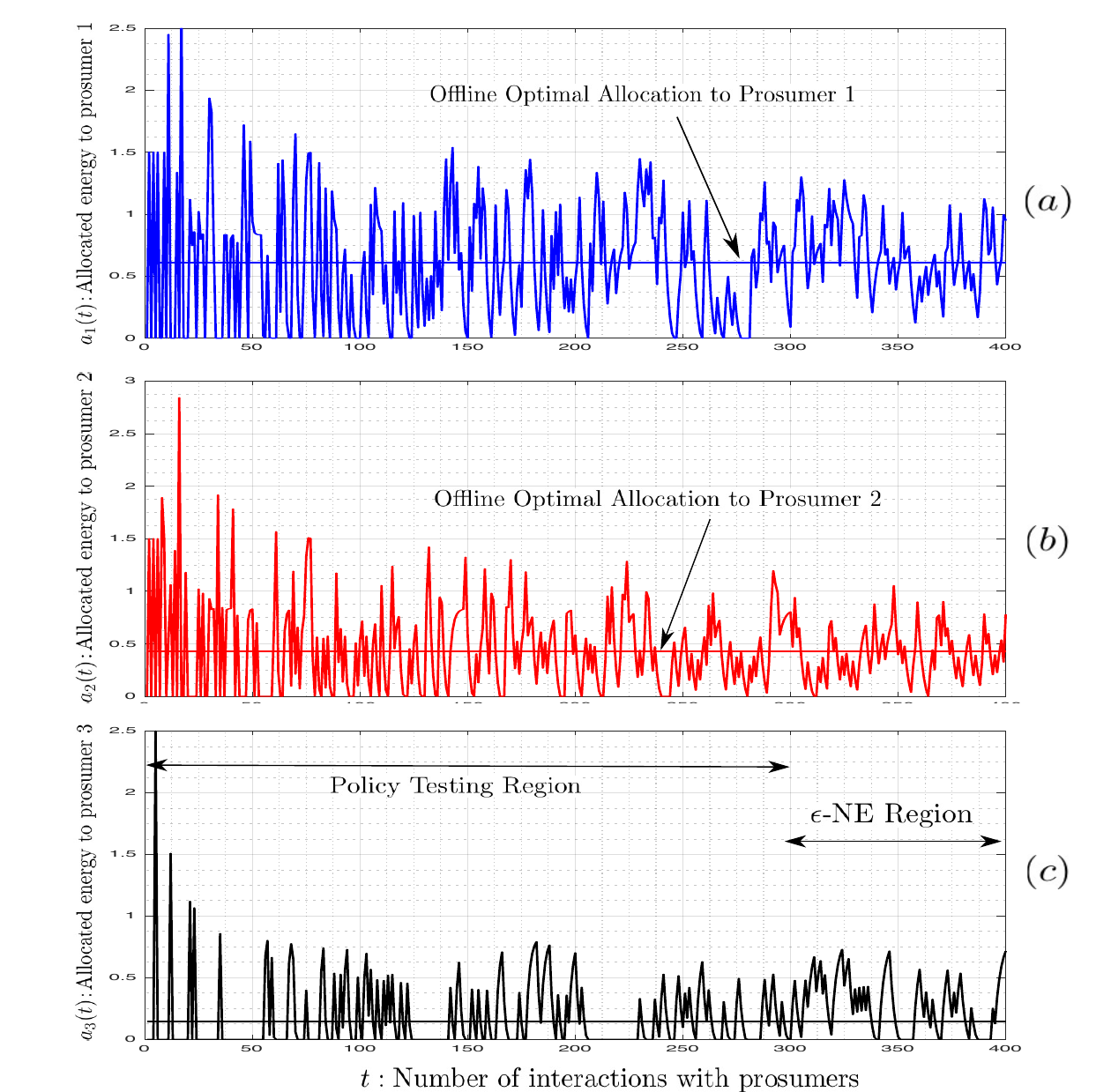} \hspace{0.4in}
\end{center}\vspace{-0.3cm}
\caption{This figure shows the amounts of online energy distribution by the utility company to each of the prosumers and their fluctuations around their optimal offline in hindsight. One can see that prosumers with less reliable renewables have higher and more unpredictable energy demands.}
\label{fig:substations}
\end{figure}

Finally, Figure \ref{fig:substations} represents the trajectories of online energy allocation by the utility company to each of the three prosumers which oscillate around their optimal offline energy allocation when the demand market is known. Moreover, the variance of fluctuations decreases as the number of interactions increases, and when the market has been stabilized in an $\epsilon$-NE. In particular, one can see that the lowest and highest amount and variance of allocated energy belong to prosumers $3$ and $1$, respectively (Figures \ref{fig:substations}-$(c)$ and \ref{fig:substations}-$(a)$, respectively). This is due to the fact that prosumer $3$ benefits more than everyone by relying on its own renewable resources (with higher expected energy generation $\mu_3=1$, and smaller variance $\sigma^2_3=1$). On the other hand, prosumer $1$ has the least reliable source of renewables (with the smallest mean and highest variance), which in turn increases its energy demand from the utility company with higher variance.

\section{Conclusion}\label{sec:conclusion}

In this paper, we have introduced a new model for energy trading in smart grids using a stochastic game framework in the demand side and an optimization problem for the utility side. We have incorporated storage devices into our formulation in order to improve energy management and have captured uncertainty of renewable resources into our grid design. Furthermore, we have formulated the prosumers' payoffs using prospect theory in order to study the real-life behavioral decisions of the prosumers. Then, we have shown that such a stochastic game admits an NE among stationary policies. In particular, we have developed a novel distributed algorithm which guarantees almost sure convergence of the prosumers' policies to an $\epsilon$-NE with very limited information sharing among them. We have also provided an online algorithm with vanishing average regret for the utility company which learns the optimal energy allocation rules over time, and have justified our results through extensive simulations.

As a future direction of research, one can consider a generalization of the stochastic game among prosumers in which there is a strong correlation among prosumers' generating energy resources. Also, in our optimization problem for the utility company, we have used regret as a performance measure where the utility company's goal is to compete with its best fixed energy allocation in hindsight. This performance measure is specially suitable in our setting since the prosumers will eventually converge to their stationary NE policies (hence in the long term what really matters for the utility company is to learn its fixed optimal allocation in the equilibrium). However, an interesting question here is to see whether there is an online allocation strategy for the utility company which achieves the same average cost as its best \emph{dynamic} allocation strategy in hindsight. 

\bibliographystyle{IEEEtran}
\bibliography{thesisrefs}

\appendices
\section{}\label{ap-static-preliminary-proofs}

\subsection{Proof of Lemma \ref{lemm:stationary}}\label{apx:lemma-stationary}
For any two states $s',s\in \{0,1,,\ldots,S_i^{\max}\}$, and $a:=(l,d)\in \{0,\ldots,L_i^{\max}\}\times \{0,\ldots,D_i^{\max}\}$, let $W^{(i)}_{s'as}$ denote the probability that the state of prosumer $i$ will change from $S_i(t)=s$ to $S_i(t+1)=s'$, by taking action $A_i(t)=a$. We note that $0<s+d-l\leq S_i^{\max}$. This is because $s+d-l$ is the amount of energy units which is left in the storage at the beginning of period $[t, t+1)$ which is when prosumer $i$ takes its action.\footnote{Note that the generated energy $G_i(t)$ will be realized and available for use only at the end of period $[t, t+1)$ or at the beginning of time $t+1$.} Since the storage capacity ranges from $0$ to $S_i^{\max}$, we have $0\leq s+d-l\leq S_i^{\max}$, and in particular, $|s'-s-d+l|\leq S_i^{\max}$. Moreover, since the event $\{G_i=s'-s-d+l\}$ is a subset of the event $\big\{\min\{[G_i+s+d-l]^+,S^{\max}_{i}\}=s'\big\}$, we can write 
\begin{align}\nonumber
W^{(i)}_{s'as}&=\mathbb{P}\{S_i(t+1)=s'|S_i(t)=s,A_i(t)=a\}\cr 
&=\mathbb{P}\Big\{\min\big\{[G_i+s+d-l]^+,S^{\max}_{i}\big\}=s'\Big\}\cr 
&\ge \mathbb{P}\{G_i=s'-s-d+l\}\cr 
&\ge \min_{|k|\leq S_i^{\max}}\mathbb{P}\{G_i=k\}= \lambda_i>0,
\end{align}
where the last inequality is due to Assumption 1. Now given a stationary policy $\Psi_i$ followed by prosumer $i$, and any two states $s,s'\in \mathcal{S}_i$, the probability that the state of prosumer $i$ is at $S_i(t+1)=s'$, given that it was at $S_i(t)=s$ before, is given by $Q_{s',s}^{\Psi_i}:=\sum_{a\in \mathcal{A}_i}\Psi_i(a|s)W^{(i)}_{s'as}\ge\sum_{a\in \mathcal{A}_i}\Psi_i(a|s)\lambda_i=\lambda_i$, where we note that $\Psi_i(a|s)$ is independent of time due to the fact that $\Psi_i$ is a stationary policy. Let $Q^{\Psi_i}$ be a $S^{\max}_i\times S_i^{\max}$ transition probability matrix with entries equal to $Q^{\Psi_i}_{s',s}$. Then the probability of being at different states at time $t$ starting from some initial state is given by a time homogeneous Markov chain with transition matrix $Q^{\Psi_i}$. Since the entries of $Q^{\Psi_i}$ are bounded below by $\lambda_i>0$, using the fundamental theorem of Markov chains there exists a unique stationary distribution $\pi^{\Psi_i}$ such that $\|(Q^{\Psi_i})^t-\boldsymbol{1}(\pi^{\Psi_i})'\|_2\leq (1-\lambda_i)^t$. This implies the result.  

\subsection{Proof of Theorem \ref{thm:existence}}\label{apx:existence}
Consider a \textit{virtual} $N$-player game in normal form where player $i$'s action set equals to the set of feasible occupation measures for prosumer $i$ in the original stochastic game:

\vspace{-0.3cm}
\begin{small}
\begin{align}\nonumber
\mathcal{M}_i\!:=\!\Big\{&\rho_i\ge 0: \!\!\sum_{(s_i,a_i)}\!\!\rho_i(s_i,a_i)(\mathbb{I}_{\{x=s_i\}}\!-\! W^i_{xa_is_i})\!=0, \forall x\in \mathcal{S}_i\Big\}.
\end{align}\end{small}In particular, an action for player $i$ in the virtual games means choosing an occupation measure for prosumer $i$. Moreover, for any joint action profile $(\rho_1,\ldots,\rho_N)$ chosen by players in the virtual game, we define the payoff for player $i$ by $J_i(\rho_i,\boldsymbol{\rho}_{-i}):=V_i(\Psi_i,\bold{\Psi}_{-i})=\sum_{(s_i,a_i)}\rho_i(s_i,a_i)K_i(s_i,a_i)$, where $\Psi_i, i=1,\ldots,n$ is given by \eqref{eq:occupation-def}. Here we note that policies $\Psi_i, i=1,\ldots,n$ are well-defined except possibly on the states where the denominator becomes zero. However, in the definition of \eqref{eq:occupation-to-policy}, if the denominator becomes zero, it means that at some state $s_i$ we have $\rho_i(s_i,a)=0, \forall a\in \mathcal{A}_i$. This in view of the definition of occupation measure \eqref{eq:occupation-def} implies that any policy which induces such occupation measure does not put any probability mass on the state $s_i$. Hence, at state $s_i$ one can define $\Psi_i(\cdot|s_i)$ to be any probability distribution over the action set $\mathcal{A}_i$, without actually changing the payoffs received by the prosumers. As a result, any two policy which are defined using the same occupation measure by \eqref{eq:occupation-def} will result in the same payoff for prosumer $i$, i.e., to the same value $V_i(\Psi_i,\bold{\Psi}_{-i})$. Hence, $J_i(\rho_i,\boldsymbol{\rho}_{-i})$ is a well-defined function.

Next, we note that if there exists a collection of $N$ occupation measures $\rho^*_i\in \mathcal{M}_i$ such that $(\rho^*_1,\ldots,\rho^*_N)$ forms a pure-strategy NE of the virtual game, then the corresponding induced stationary policies given by $\Psi^*_i(a_i|s_i)=\frac{\rho^*_i(s_i,a_i)}{\sum_{a_j\in \mathcal{A}_i}\rho^*_i(s_i,a_j)}$ will be a stationary NE policy for the original stochastic game. To see this more clearly, let $(\rho^*_i,\boldsymbol{\rho}^*_{-i})$ be a pure-strategy NE of the virtual game. This means that 
\begin{align}\nonumber
&V_i(\Psi^*_i,\boldsymbol{\Psi}^*_{-i})=J_i(\rho^*_i,\boldsymbol{\rho}^*_{-i})=\max_{\rho_i\in \mathcal{M}_i}J_i(\rho_i,\boldsymbol{\rho}^*_{-i})\cr 
&\qquad= \max_{\rho_i\in \mathcal{M}_i}\sum_{(s_i,a_i)}\rho_i(s_i,a_i)K^*_i(s_i,a_i)=\max_{\theta_i}V_i(\theta_i,\boldsymbol{\Psi}^*_{-i}),
\end{align} 
where $K^*_i(s_i,a_i)\!:=\!\!\!\!\!\sum\limits_{({\rm \boldsymbol{s}}_{-i},{\rm \boldsymbol{a}}_{-i})}\!\!\!\!\!\!w_i\big(\!\!\prod\limits_{j\neq i}\Psi^*_j(a_j|s_j)\pi^{\Psi^*_j}(s_j)\big)v_i\big(U_i(\boldsymbol{a})\big)$, $\theta_i$ is an arbitrary policy for prosumer $i$, and the last equality holds because $\max_{\rho_i\in \mathcal{M}_i}\sum_{(s_i,a_i)}\rho_i(s_i,a_i)K_i(s_i,a_i)$ is exactly the linear program \eqref{eq:LP} describing the optimal policy (best response) of prosumer $i$ for given choice of stationary policies $\bold{\Psi}^*_{-i}$ of other prosumers.

Therefore, we are only left to show that the virtual game admits a pure-strategy NE. For this purpose, we benefit from the following Lemma from \cite{rosen1965existence}:   

\begin{lemma}\cite[Theorem 1]{rosen1965existence}\label{lemm:contractable}
\normalfont Consider a \emph{concave} game in normal form where each player $i$, chooses an action $m_i\in \mathcal{M}_i$, where $\mathcal{M}_i$ is a closed convex bounded set, and obtains a payoff $J_i(m_i,{\rm \boldsymbol{m}}_{-i})$. Assume that $J_i(m_i,{\rm \boldsymbol{m}}_{-i})$ is a concave function of $m_i$ for every arbitrary but fixed ${\rm \boldsymbol{m}}_{-i}$, and continuous function of ${\rm \boldsymbol{m}}, \forall {\rm \boldsymbol{m}}\in \mathcal{M}_1\times\ldots\times\mathcal{M}_N$. Then the game admits a pure-strategy NE. 
\end{lemma}

To complete the proof, we show that the virtual game is indeed a so called \textit{concave} game which satisfies the conditions of Lemma \ref{lemm:contractable}. Clearly the action sets $\mathcal{M}_i$ are closed convex polyhedron which are determined by a set of linear constraints. Moreover, for any arbitrary but fixed $\boldsymbol{\rho}_{-i}\in \mathcal{M}_{-i}$, the payoff of player $i$ is given by $J_i(\rho_i,\boldsymbol{\rho}_{-i})=\sum_{(s_i,a_i)}\rho(s_i,a_i)K_i(s_i,a_i)$, which is a linear function of $\rho_i$, and hence a concave function of $\rho_i$. Moreover, using Lemma \ref{lemm:continuity-stationary} and by continuity of $w_i(\cdot)$, the term $K_i(s_i,a_i)=\sum_{({\rm \bold{s}}_{-i},{\rm \bold{a}}_{-i})}w_i\Big(\prod_{j\neq i}\Psi_j(a_j|s_j)\pi^{\Psi_j}(s_j)\Big)v_i\Big(U_i(\bold{a})\Big)$ is a continuous function of $\bold{\Psi}_{-i}$. Since $\bold{\Psi}_{-i}$ itself is a continuous function of $\boldsymbol{\rho}_{-i}$ (recall that for all $k$, $\Psi_k(a_k|s_k)=\frac{\rho_k(s_k,a_k)}{\sum_{a_j\in \mathcal{A}_k}\rho_k(s_k,a_j)}$ which is a continuous function of $\rho_k$), therefore $J_i(\rho):=\sum_{(s_i,a_i)}\rho(s_i,a_i)K_i(s_i,a_i)$ is a continuous function of $\rho=(\rho_i,\boldsymbol{\rho}_{-i})$. Appealing to the result of Lemma \ref{lemm:contractable} one can see that the virtual game admits a pure-strategy NE. This completes the proof.

\subsection{Proof of Lemma \ref{lemm:epsilon-value-error}}\label{apx:epsilon-value-error}
First we note that for a joint stationary policy $\bold{\Psi}$, we have 

\vspace{-0.3cm}
\begin{small}
\begin{align}\nonumber
V_i(\bold{\Psi})\!=\!\sum_{({\rm \bold{s}},{\rm \bold{a}})}\pi^{\Psi_i}(s_i)\Psi_i(a_i|s_i)w_i\big(\prod\limits_{j\neq i}\Psi_j(a_j|s_j)\pi^{\Psi_j}(s_j)\big)v_i\big(U_i(\bold{a})\big).
\end{align}\end{small}Since $w_i(\cdot)$ is continuous over the compact interval $[0,1]$, it is uniformly continuous. Thus for any $\epsilon>0$, there exists a positive constant $\delta_i$ such that if $|x-y|<\delta_i$, then $|w_i(x)-w_i(y)|<\epsilon$. Let $\delta=\min\delta_i$, $\epsilon_1=\min\{\epsilon,\frac{\delta}{N^2}\}$, and $T=\frac{\ln(\epsilon_1)}{\ln(1-\lambda)}$, where $\lambda=\min\lambda_j$. Using Lemma \ref{lemm:stationary}, if $t>T$, we have $\max_{s_{j}\in \mathcal{S}_{j}}\left|\mathbb{P}^{\Psi_j}\{S_j(t)=s_j\}-\pi^{\Psi_j}(s_j)\right|<\epsilon_1$ for all prosumers $j$. For simplicity of notation let us define $\mathbb{P}_t^{\Psi_j}(s_j,a_j):=\mathbb{P}^{\Psi_j}\{S_j(t)=s_j,A_j(t)=a_j\}$. We can write
\begin{align}\nonumber
|\mathbb{P}_t^{\Psi_j}(s_j,a_j)&-\pi^{\Psi_j}(s_j)\Psi_j(a_j|s_j)|\cr 
&=\Psi_j(a_j|s_j)\left|\mathbb{P}^{\Psi_j}\{S_j(t)=s_j\}-\pi^{\Psi_j}(s_j)\right|\cr 
&<\Psi_j(a_j|s_j)\epsilon_1, \ \ \forall j. 
\end{align}Using this relation, we can write  

\vspace{-0.3cm}
\begin{small}
\begin{align}
&\Big|\prod_{j\neq i}\mathbb{P}_t^{\Psi_j}(s_j,a_j)-\prod_{j\neq i}\pi^{\Psi_j}(s_j)\Psi_j(a_j|s_j)\Big|\cr
&\!\leq\! \Big|\!\prod_{j\neq i}\!\left(\!\pi^{\Psi_j}(s_j)\Psi_j(a_j|s_j)\!+\!\Psi_j(a_j|s_j)\epsilon_1\!\right)\!-\!\prod_{j\neq i}\!\pi^{\Psi_j}(s_j)\Psi_j(a_j|s_j)\Big|\cr 
&\!=\!\left(\!\prod_{j\neq i}\Psi_j(a_j|s_j)\!\right)\!\Big|\prod_{j\neq i}\left(\pi^{\Psi_j}(s_j)\!+\!\epsilon_1\right)\!-\!\prod_{j\neq i}\pi^{\Psi_j}(s_j)\Big|\cr
&\!\leq\!\Big|\prod_{j\neq i}\!\left(\!\pi^{\Psi_j}(s_j)\!+\!\epsilon_1\!\right)\!-\!\prod_{j\neq i}\!\pi^{\Psi_j}(s_j)\Big|\!=\!\epsilon_1 \frac{d}{dx}\!\prod_{j\neq i}\!\left(\!\pi^{\Psi_j}(s_j)\!+\!x\!\right)_{\!|_{x\in[0,\epsilon_1]}}\cr 
&\leq (N-1)\epsilon_1 (1+\epsilon_1)^{N-2}<N\times \frac{\delta}{N^2}\times (1+\frac{\delta}{N^2})^N<\delta, 
\end{align}\end{small}where the second inequality is because $\Psi_j(a_j|s_j)\in [0,1], \forall j$, and the last equality is by mean value theorem. Now let $\Delta_1$ be a uniform upper bound on the value function of the instantaneous payoffs, i.e., $\Delta_1=\max_{i,{\bold{a}}}v_i\big(U_i(\bold{a})\big)$, and $\Delta_2$ be a uniform upper bound on the weight functions, i.e., $\Delta_2=\max_{i, x\in[0,1]}w_i(x)$. Denoting the set of all states and actions by $\mathcal{S}$ and $\mathcal{A}$, respectively, we have 

\vspace{-0.3cm}
\begin{footnotesize}
\begin{align}\nonumber
&|V^{T(\epsilon)}_i(\boldsymbol{\Psi})-V_i(\boldsymbol{\Psi})|=\cr 
&=\frac{1}{T(\epsilon)}\Big|\sum_{({\rm \bold{s}},{\rm \bold{a}})}\sum_{t=1}^{T(\epsilon)}\Big[\mathbb{P}_t^{\Psi_i}(s_i,a_i)w_i(\prod_{j\neq i}\mathbb{P}_t^{\Psi_{j}}(s_j,a_j))\cr 
&\qquad\qquad-\pi^{\Psi_i}(s_i)\Psi_i(a_i|s_i)w_i\big(\!\prod\limits_{j\neq i}\Psi_j(a_j|s_j)\pi^{\Psi_j}(s_j)\big)\Big]\!v_i\big(U_i(\bold{a})\big)\Big|\cr 
&\leq\frac{\Delta_1}{T(\epsilon)}\sum_{({\rm \bold{s}},{\rm \bold{a}})}\sum_{t=1}^{T}\Delta_2\!+\!\frac{\Delta_1}{T(\epsilon)}\Big|\!\sum_{({\rm \bold{s}},{\rm \bold{a}})}\!\sum_{t=T}^{T(\epsilon)}\Big[\mathbb{P}_t^{\Psi_i}(s_i,a_i)w_i(\prod_{j\neq i}\mathbb{P}_t^{\Psi_{j}}(s_j,a_j))\cr 
&\qquad\qquad\qquad\qquad\qquad-\pi^{\Psi_i}(s_i)\Psi_i(a_i|s_i)w_i\big(\!\prod\limits_{j\neq i}\!\Psi_j(a_j|s_j)\pi^{\Psi_j}(s_j)\big)\Big]\Big|\cr
&\leq\frac{\Delta_1}{T(\epsilon)}\!\!\sum_{({\rm \bold{s}},{\rm \bold{a}})}\!\sum_{t=1}^{T}\Delta_2\!+\!\frac{\Delta_1}{T(\epsilon)}\!\!\sum_{({\rm \bold{s}},{\rm \bold{a}})}\!\sum_{t=T}^{T(\epsilon)}\Big|\mathbb{P}_t^{\Psi_i}(s_i,a_i)\!-\!\pi^{\Psi_i}(s_i)\Psi_i(a_i|s_i)\Big|\Delta_2\cr
&+\frac{\Delta_1}{T(\epsilon)}\sum_{({\rm \bold{s}},{\rm \bold{a}})}\sum_{t=T}^{T(\epsilon)}\Big|w_i(\prod_{j\neq i}\mathbb{P}_t^{\Psi_{j}}(s_j,a_j))\!-\!w_i\big(\prod\limits_{j\neq i}\Psi_j(a_j|s_j)\pi^{\Psi_j}(s_j)\big)\Big|\cr
&\leq \frac{T\Delta_1\Delta_2|\mathcal{S}||\mathcal{A}|}{T(\epsilon)}\!+\!\frac{(T(\epsilon)\!-\!T)\Delta_1\Delta_2|\mathcal{S}||\mathcal{A}|\epsilon}{T(\epsilon)}\!+\!\frac{(T(\epsilon)\!-\!T)\Delta_1|\mathcal{S}||\mathcal{A}|\epsilon}{T(\epsilon)},
\end{align}\end{footnotesize}where the first inequality is obtained by slitting the sum over two intervals $[1,T]\cup [T,T(\epsilon)]$, and upper bounding each of them in terms of $\Delta_1$ and $\Delta_2$. The second inequality hods by rearranging the terms and using triangle inequality, and the last inequality is valid by Lemma \ref{lemm:stationary} and definitions of $\Delta_1$ and $\Delta_2$. Finally, by choosing $T(\epsilon)\ge \frac{3T\Delta_1\Delta_2|\mathcal{S}||\mathcal{A}|}{\epsilon}$ (recall that $T=\frac{\ln(\min\{\epsilon,\frac{\delta}{N^2}\})}{\ln(1-\lambda)}$) the right hand side of the above relation will be at most $\epsilon$, which completes the proof.

\subsection{Proof of Theorem \ref{thm-alg-convergence}}\label{apx:alg-convergence}
First we argue that if $q_i^{m}=1, \forall i$, then the policy which is followed last by the prosumers, i.e., $\bold{\Psi}(m-1)$ must constitute an $\epsilon$-NE, which by the rule of the algorithm it must be played forever. Let us consider an arbitrary but fixed $i$. When $q_i^{m}=1$, this means that $\hat{V}_i^{(m)}>\max\limits_{k=1,\ldots,r}\hat{V}_{i_k}^{(m)}-\epsilon$. Since in the sub-intervals $Z_{(i-1)r+1},\ldots,Z_{ir}$, only the $i$th prosumer is sampling the policies corresponding to its vertex points, we have $\bold{\Psi}_{-i}(Z_{(i-1)r+k})=\bold{\Psi}_{-i}(m-1)$. We can write 

\vspace{-0.3cm}
\begin{small}
\begin{align}\nonumber
&|\hat{V}_{i_k}^{(m)}-V_i(\bold{\Psi}_{-i}(m-1),\Psi_{i_k})|\cr 
&\qquad=|V^{T(\epsilon)}_i(\bold{\Psi}_{-i}(Z_{(i-1)r+k}),\Psi_{i_k})-V_i(\bold{\Psi}_{-i}(m-1),\Psi_{i_k})|\cr 
&\qquad=|V^{T(\epsilon)}_i(\bold{\Psi}_{-i}(m-1),\Psi_{i_k})-V_i(\bold{\Psi}_{-i}(m-1),\Psi_{i_k})|< \epsilon,
\end{align}\end{small}where the last inequality is due to Lemma \ref{lemm:epsilon-value-error}. Similarly, in the last sub-interval prosumers are following $\bold{\Psi}(m-1)$. Therefore, $|\hat{V}_i^{(m)}-V_i(\bold{\Psi}(m-1))|<\epsilon$. Now we have

\vspace{-0.3cm}
\begin{small}
\begin{align}\label{eq:epsilon-equilibrium}
V_i(\bold{\Psi}(m-1))&>\hat{V}_i^{(m)}-\epsilon>\max\limits_{k=1,\ldots,r}\hat{V}_{i_k}^{(m)}-2\epsilon\cr 
&>\max\limits_{k=1,\ldots,r}V_i(\bold{\Psi}_{-i}(m-1),\Psi_{i_k})-3\epsilon\cr 
&=\max\limits_{\Psi_i}V_i(\bold{\Psi}_{-i}(m-1),\Psi_i)-3\epsilon
\end{align}\end{small}where the last equality holds because the maximum payoff of the $i$th prosumer is obtained by following a policy corresponding to one of the extreme points of player $i$'s actions in the virtual game. Since \eqref{eq:epsilon-equilibrium} holds for all $i$, this implies that $\bold{\Psi}(m-1)$ must be a $3\epsilon$-NE (or simply $\epsilon$-NE by rescaling $\epsilon$ to $\frac{\epsilon}{3}$ in all the above analysis). 

On the other hand, if $\boldsymbol{\rho}(m-1)$ is an $\epsilon$-NE for the virtual game, then $\bold{\Psi}(m-1)$ is an $\epsilon$-Nash policy for the original game, in which case using similar argument as \eqref{eq:epsilon-equilibrium} one can show that $\hat{V}_i^{(m)}>\max\limits_{k=1,\ldots,r}\hat{V}_{i_k}^{(m)}-\epsilon, \forall i$, and hence $q^{(m)}_i=1, \forall i$. In other words, once the random selection of occupation measures at the beginning of some long interval $[(m-1)(Nr+1)T(\epsilon)+1, m(Nr+1)T(\epsilon))$ forms an $\epsilon$-NE, then the algorithm will stuck there and prosumers will continue playing that $\epsilon$-Nash policy forever.  

Finally, we need to show that almost surely, there will be a time such that $q_i^{m}=1, \forall i$. Using Theorem \ref{thm:existence} we know that the virtual game has at least one pure-strategy NE denoted by $\boldsymbol{\rho}^*$. Since $J_i(\boldsymbol{\rho})$ is a continuous function of $\boldsymbol{\rho}\in \mathcal{M}_1\times\ldots\times\mathcal{M}_N$, for any $\epsilon>0$, there is a $\delta>0$ such that if $|\rho_i-\rho^*_i|\leq \delta, \forall i$, then $|J_i(\boldsymbol{\rho})-J_i(\boldsymbol{\rho}^*)|<\epsilon, \forall i$. In other words, all the occupation measures profiles $\boldsymbol{\rho}$ satisfying $|\rho_i-\rho^*_i|\leq \delta, \forall i$, will form an $\epsilon$-NE for the virtual game (and hence, their corresponding policies are $\epsilon$-Nash policy for the original game). Since, at the beginning of each longer period $m=1,2,\ldots$, each player $i$ chooses an occupation measure $\rho_i(m-1)$ from $\mathcal{M}_i$, uniformly at random and independently from others, thus the probability that the randomly selected profile $\boldsymbol{\rho}(m-1)$ is an $\epsilon$-NE is at least $\delta^N$. Let $M$ be a random variable denoting the first integer $m$ such that $\boldsymbol{\rho}(m-1)$ is an $\epsilon$-NE. We have $\mathbb{P}\{M\ge m\}<(1-\delta^{N})^{m}$. Since $\sum_{m=1}^{\infty}\mathbb{P}\{M\ge m\}<\infty$, Borel-Cantelli lemma implies that almost surely $M<\infty$, i.e., the algorithm finds an $\epsilon$-NE policy and uses it forever.

\subsection{Proof of Theorem \ref{thm:online-algorithm-utility}}\label{apx:online-algorithm-utility}

For any arbitrary but fixed demand profile $\boldsymbol{D}(t)$, one can easily see that the cost of the utility company given by \eqref{eq:company-cost} is a convex function of energy allocation profile $\boldsymbol{e}(t)$. This is simply because for fixed $\boldsymbol{D}(t)$,  $C(\boldsymbol{D}(t), \boldsymbol{e}(t))$ is a quadratic function in terms of $\boldsymbol{e}(t)$. Moreover, the action set of utility company $\mathcal{E}$ given by \eqref{eq:action-set-utility-company} is a closed convex set, and we have $\max\limits_{x,y\in \mathcal{E}}\|x-y\|=E_{\max}<\infty$. In addition, given fixed $\boldsymbol{D}(t)$, the function $C(\boldsymbol{D}(t),\boldsymbol{e}(t))$ is differentiable with respect to $\boldsymbol{e}(t)$, and we have $\|\nabla_{\boldsymbol{e}(t)} C(\boldsymbol{D}(t),\boldsymbol{e}(t))\|^2\!=\!\sum_{\ell=1}^{K}\Big(\beta\!-2\gamma\!\!\sum_{j\in \mathcal{B}_{\ell}}\!\!D_j(t)\!+2\gamma e_{\ell}(t)\Big)^2\leq K(\beta+2\gamma K D_{\max}+2\gamma E_{\max})^2$, which shows that $\max_{\boldsymbol{e}(t)\in \mathcal{E}}\|\nabla_{\boldsymbol{e}(t)} C(\boldsymbol{D}(t),\boldsymbol{e}(t))\|$ is uniformly bounded above for all $t$. 

Next, we state the following lemma from online convex optimization whose proof can be found in \cite[Theorem 1]{zinkevich2003online}: 
\begin{lemma}\label{lemm:online-convex}
\normalfont Let $\mathcal{C}$ be a closed convex set. For $t=1,2,\ldots$, consider a sequence of decision points $x_t\!\in\!\mathcal{C}$, and a sequence of convex functions $h_t(\cdot)\!:\!\mathcal{C}\to\mathbb{R}$, such that $x_{t}=\Pi_{\mathcal{C}}[x_{t-1}-\gamma_{t-1}\nabla h_{t-1}(x_{t-1})]$. If $F:=\max\limits_{x,y\in \mathcal{C}}\|x-y\|<\infty$, and $H:=\max\limits_{x\in \mathcal{C}, t=1,2,\ldots}\|\nabla h_t(x)\|<\infty$, then by choosing $\gamma_t:=\frac{1}{\sqrt{t}}$, we have $\sum_{t=1}^{T}h_t(x_t)-\min\limits_{x\in\mathcal{C}}\sum_{t=1}^{T}h_t(x)\leq (\frac{1}{2}F^2+H^2)\sqrt{T}$.
\end{lemma} 

To finish the proof, one can imagine that at each time step $t$, the utility company takes an action in the closed convex set $\mathcal{E}$, and the prosumers collaboratively with nature choose the demand $\boldsymbol{D}(t)$, which in turn determines a convex function $h_t(x):=C(\boldsymbol{D}(t),x)$. Now if the utility company selects its allocation energy $\boldsymbol{e}(t)$ to different substations based on \eqref{eq:online-alg-2}, then appealing to Lemma \ref{lemm:online-convex} one can see that the average regret $\frac{R_T}{T}$ of the utility company is bounded by $O(\frac{1}{\sqrt{T}})$.

\subsection{Continuity of Stationary Distribution}\label{apx:continuity-stationary}
\begin{lemma}\label{lemm:continuity-stationary}
\normalfont Let $\pi^{\Psi_i}$ denote the unique stationary distribution corresponding to the Markov chain with transition probabilities $Q_{s',s}^{\Psi_i}:=\sum_{a\in \mathcal{A}_i}\Psi_i(a|s)W^{(i)}_{s'as}$ where $W^{(i)}_{s'as}$ are constant numbers and $\Psi_i$ is a stationary policy. Then $\pi^{\Psi_i}$ is a continuous function of $\Psi_i$.
\end{lemma}
\begin{proof}
Since the transition weights $Q_{s',s}^{\Psi_i}$ are linear functions of $\Psi_i$, the transition matrix $Q^{\Psi_i}=(Q_{s',s}^{\Psi_i})_{s',s}$ is a continuous function of $\Psi_i$. Now assume $\{\Psi_{i}(k)\}_{k=1}^{\infty}$ be a sequence of stationary policies converging to $\Psi_i$, and denote the unique stationary policies of $Q^{\Psi_i(k)}$ by $\pi^{\Psi_i(k)}$. Since $\{\pi^{\Psi_i(k)}\}_{k=1}^{\infty}$ is a sequence in the product of compact probability simplexes, it has an accumulation point $\pi^*$ with a sub-sequence converging to it. With some abuse of notation let us denote this sub-sequence again by $\{\pi^{\Psi_i(k)}\}_{k=1}^{\infty}$. We have
\begin{align}\nonumber
\pi^{*}=\lim_{k\to \infty}\pi^{\Psi_i(k)}=\lim_{k\to \infty}Q^{\Psi_{i}(k)}\pi^{\Psi_{i}(k)}=Q^{\Psi_{i}}\pi^{*}.
\end{align} 
This means that $\pi^{*}$ is also a stationary distribution for $Q^{\Psi_i}$. Since we have assumed that $Q^{\Psi_i}$ has a unique stationary distribution $\pi^{\Psi_i}$, this implies that $\pi^{\Psi_i}=\pi^{*}$. Thus $\pi^{\Psi_i}$ is the only accumulation point of $\pi^{\Psi_i(k)}$, which shows that $\pi^{\Psi_i}$ is a continuous function of $\Psi_i$.  
\end{proof}

\vspace{-1cm}
\begin{biographynophoto}
{\bf S. Rasoul Etesami} (S'12, M'16) received his Ph.D. degree in Electrical and Computer Engineering in 2015 from University of Illinois at Urbana-Champaign. He was a postdoctoral research fellow at Princeton University until May 2017. His research interests include social and distributed networks, networked games, smart grids, and algorithm design. \vspace{-1.1cm}
\end{biographynophoto}

 
\begin{biographynophoto}
{\bf Walid Saad} (S'07, M'10, SM’15) received his Ph.D degree from the University of Oslo in 2010. Currently,  he is an Associate Professor at the Department of Electrical and Computer Engineering at Virginia Tech, where he leads the Network Science, Wireless, and Security (NetSciWiS) laboratory, within the Wireless@VT research group. His  research interests include wireless networks, game theory, cybersecurity, unmanned aerial vehicles, and cyber-physical systems. Dr. Saad is the recipient of the NSF CAREER award in 2013, the AFOSR summer faculty fellowship in 2014, and the Young Investigator Award from the Office of Naval Research (ONR) in 2015. He was the author/co-author of six conference best paper awards at WiOpt in 2009, ICIMP in 2010, IEEE WCNC in 2012,  IEEE PIMRC in 2015, IEEE SmartGridComm in 2015, and EuCNC in 2017. He is the recipient of the 2015 Fred W. Ellersick Prize from the IEEE Communications Society. In 2017, Dr. Saad was named College of Engineering Faculty Fellow at Virginia Tech. He currently serves as an editor for the IEEE Transactions on Wireless Communications, IEEE Transactions on Communications, and Transactions on Information Forensics and Security.\vspace{-1.1cm}
\end{biographynophoto}

\begin{biographynophoto}
{\bf Narayan B. Mandayam} (S'89, M'94, SM'99, F'09)
received the M.S. and Ph.D. degrees from Rice University in 1991
and 1994, respectively, all in electrical engineering.
From 1994 to 1996, he was a Research Associate
at the Wireless Information Network Laboratory
(WINLAB), Rutgers University, New Brunswick, NJ,
USA, before joining the faculty of the Electrical and
Computer Engineering department at Rutgers where
he is currently a Distinguished Professor. He also
serves as Associate Director at WINLAB. He was a visiting faculty fellow
in the Department of Electrical Engineering, Princeton University, Princeton,
NJ, USA, in 2002 and a visiting faculty at the Indian Institute of Science,
Bengaluru, India, in 2003. His research interests are in various aspects of
wireless data transmission, modeling social knowledge creation on the internet, game
theory, communications and networking, signal processing, among many others. Dr. Mandayam is a co-recipient of the
2015 IEEE Communications Society Advances in Communications Award
for his seminal work on power control and pricing, the 2014 IEEE Donald
G. Fink Award for his IEEE Proceedings paper titled “‘Frontiers of Wireless
and Mobile Communications’” and the 2009 Fred W. Ellersick Prize from
the IEEE Communications Society for his work on dynamic spectrum access
models and spectrum policy. He is also a recipient of the Peter D. Cherasia
Faculty Scholar Award from Rutgers University (2010), the National Science
Foundation CAREER Award (1998) and the Institute Silver Medal from
the Indian Institute of Technology (1989). He is a coauthor of the books:
Principles of Cognitive Radio (Cambridge University Press, 2012) and Wireless
Networks: Multiuser Detection in Cross-Layer Design (Springer, 2004). He has
served as an Editor for the journals IEEE Communication Letters and IEEE
Transactions on Wireless Communications. He has also served as a guest editor
of the IEEE JSAC Special Issues on Adaptive, Spectrum Agile and Cognitive
Radio Networks (2007) and Game Theory in Communication Systems
(2008).
\end{biographynophoto}
\vspace{-1.1cm}

\begin{biographynophoto}
{\bf H. Vincent Poor} (S'72, M'77, SM'82, F'87) received the Ph.D. degree in electrical engineering and computer science from Princeton University in 1977.  From 1977 until 1990, he was on the faculty of the University of Illinois at Urbana-Champaign. Since 1990 he has been on the faculty at Princeton, the Michael Henry Strater University Professor of Electrical Engineering. During 2006-16 he served as Dean of Princeton's School of Engineering and Applied Science. He has also held visiting appointments at several other institutions, most recently at Berkeley and Cambridge. His research interests are in the areas of information theory, stochastic analysis and statistical signal processing, and their applications in wireless networks and related fields such as smart grid. Among his publications in these areas is the recent book \textit{Mechanisms and Games for Dynamic Spectrum Allocation} (Cambridge University Press, 2014).

Dr. Poor is a member of the National Academy of Engineering and the National Academy of Sciences, and is a foreign member of the Royal Society. He is also a fellow of the American Academy of Arts and Sciences and of other national and international academies. He received a Guggenheim Fellowship in 2002 and the IEEE Education Medal in 2005. Recent recognition of his work includes the the 2016 John Fritz Medal, the 2017 IEEE Alexander Graham Bell Medal, Honorary Professorships at Peking University and Tsinghua University, both conferred in 2016, and a D.Sc. \textit{honoris causa} from Syracuse University awarded in 2017.

\end{biographynophoto}

\end{document}